\documentclass[preprint]{acmart}
\renewcommand\footnotetextcopyrightpermission[1]{} 
\acmJournal{PACMPL}
\acmVolume{1}
\acmNumber{CONF} 
\acmArticle{1}
\acmYear{2018}
\acmMonth{1}
\acmDOI{} 
\startPage{1}

\setcopyright{none}

\usepackage{booktabs}   
\usepackage{subcaption} 

\usepackage{amsmath}
	\makeatletter
	\newcommand{\leqnos}{\tagsleft@true\let\veqno\@@leqno}
	\newcommand{\reqnos}{\tagsleft@false\let\veqno\@@eqno}
	\reqnos
	\makeatother 
\usepackage{amssymb}
\usepackage{stmaryrd}
\usepackage{cmll}
\usepackage{proof}
\usepackage{subcaption}
\usepackage{tikz}
\usetikzlibrary{shapes.geometric}

\usepackage{listings}
\usepackage{xcolor}
\definecolor{dkgreen}{rgb}{0,0.6,0}
\definecolor{gray}{rgb}{0.5,0.5,0.5}
\definecolor{mauve}{rgb}{0.58,0,0.82}
\definecolor{hilight}{RGB}{122,86,0}

\newcommand{\lintype}{\tyR}
\newcommand{\exptype}[1]{\oc #1}
\newcommand{\gentype}[1]{(\oc) #1} 

\def\ie{{\it i.e.}}
\def\eg{{\it e.g.}}
\def\cF{\mathcal F}

\newcommand{\Der}{\mathbf D}
\newcommand{\opFwd}{\overrightarrow{\mathbf D}}
\newcommand{\Fwd}[2][]{\opFwd_{#1}(#2)}
\newcommand{\Derr}{\mathbf D_{\mathrm r}}
\newcommand{\opRev}{\overleftarrow{\mathbf D}}
\newcommand{\RevDer}[1]{\opRev(#1)}
\newcommand{\Terms}[1][]{\Lambda(#1)}
\newcommand{\DualTerms}[1][]{\Lambda_\bot(#1)}
\newcommand{\Nat}{\mathbb N}
\newcommand{\Real}{\mathbb R}

\newcommand{\gsep}{\mathrel{|}}
\newcommand{\pair}[1]{\left\langle #1\right\rangle}
\newcommand{\unpair}[4]{#1\esub{\pair{#2,#3}}{#4}}
\newcommand{\isub}[2]{\!\left\{#1/#2\right\}}
\newcommand{\esub}[2]{[#1:=#2]}

\newcommand{\size}[1]{\left|#1\right|}
\newcommand{\bisize}[1]{\left\|#1\right\|}

\newcommand{\metaCtxt}{\mathsf C}
\newcommand{\metaCtxtp}{\mathsf D}
\newcommand{\metaSubCtxt}{\alpha}
\newcommand{\metaSubCtxtp}{\beta}
\newcommand{\ctxt}[2][\metaCtxt]{#1\!\left\{#2\right\}\!}
\newcommand{\subctxt}[2][\metaSubCtxt]{#2 #1}
\newcommand{\red}[1][]{\xrightarrow{#1}}
\newcommand{\reds}[2][\ast]{\mathrel{\xrightarrow{#2}\!\!{}^{#1}}}
\newcommand{\fred}{\red}
\newcommand{\freds}[1][\ast]{\reds[#1]{}}
\newcommand{\fv}[1]{\mathsf{fv}(#1)}

\newcommand{\tyR}{\mathsf R}

\newcommand{\revder}[2][]{\overleftarrow{\mathbf D}_{#1}(#2)}
\newcommand{\revderSymbol}[1][]{\overleftarrow{\mathbf D}_{#1}}
\newcommand{\cat}[1]{\mathbf{#1}}

\newcommand{\sem}[1]{\left\llbracket #1\right\rrbracket}

\newcommand{\seq}[1]{\mathbf{#1}}
\newcommand{\num}[1]{\underline{#1}}

\newcommand{\dual}[1][]{\tyR^{\bot_{#1}}}
\newcommand{\bp}[2][\seq x,a]{\mathbf{bp}_{#1}(#2)}
\newcommand{\tyList}[2]{\mathsf{List}(#1,#2)}

\newcommand{\dualvar}{^\ast}

\newcommand{\Mult}[2]{#1\ensuremath{\cdot}#2} 

\newcommand{\reflemma}[1]{Lemma~\ref{lemma:#1}}
\newcommand{\refcor}[1]{Corollary~\ref{cor:#1}}

\newcommand{\reffig}[1]{Fig.~\ref{fig:#1}}
\newcommand{\refprop}[1]{Proposition~\ref{prop:#1}}
\newcommand{\refsect}[1]{Sect.~\ref{sect:#1}}
\newcommand{\refsubsect}[1]{Sect.~\ref{subsect:#1}}
\newcommand{\reftab}[1]{Tab.~\ref{table:#1}}

\definecolor{blue(pigment)}{rgb}{0.2, 0.2, 0.6}
\definecolor{britishracinggreen}{rgb}{0.0, 0.26, 0.15}
\definecolor{burgundy}{rgb}{0.5, 0.0, 0.13}
\newcommand{\cval}[1]{\textcolor{blue(pigment)}{#1}}
\newcommand{\cder}[1]{\textcolor{britishracinggreen}{#1}}
\newcommand{\crder}[1]{\textcolor{burgundy}{#1}}

\def\scalefact{0.85}
\newcommand{\znode}[5][black]{\path (#3,#4) node(#2) [circle,draw,color=#1] {#5};}
\newcommand{\zunedge}[6][black]{%
\begin{scope}
	\path (#2,#3) node(this) [inner sep=0pt,triangle,draw,color=#1] {#4};
	\draw[->,color=#1] (#5) -- (this.west);
	\draw[->,color=#1] (this.east) -- (#6);
\end{scope}}
\newcommand{\zbiedge}[7][black]{%
\begin{scope}
	\path (#2,#3) node(this) [inner sep=0pt,triangle,draw,color=#1] {#4};
	\draw[->,color=#1] (#5) -- (this);
	\draw[->,color=#1] (#6) -- (this);
	\draw[->,color=#1] (this.east) -- (#7);
\end{scope}}
\newcommand{\zedge}[5][black]{\path (#3,#4) node(#2) [inner sep=0pt,triangle,draw,color=#1] {#5};}

\begin{document}

\title[Backprop in Lambda-Calculus]{Backpropagation in the Simply Typed Lambda-Calculus with Linear Negation}         



\author{Alo\"\i s  Brunel}
\affiliation{
  \institution{Deepomatic}            
  \country{France}                    
}
\email{alois.brunel@gmail.com}          

\author{Damiano Mazza}
\affiliation{
  \institution{CNRS, UMR 7030, LIPN, Universit\'e Sorbonne Paris Nord}            
  \country{France}                    
}
\email{Damiano.Mazza@lipn.univ-paris13.fr}          

\author{Michele Pagani}
\affiliation{
  \institution{IRIF UMR 8243, Universit\'e de Paris, CNRS}            
  \country{France}                    
}
\email{pagani@irif.fr}          

\begin{abstract}

Backpropagation is a classic automatic differentiation algorithm computing the gradient of functions specified by a certain class of simple, first-order programs, called computational graphs.  It is a fundamental tool in several fields, most notably machine learning, where it is the key for efficiently training (deep) neural networks.  Recent years have witnessed the quick growth of a research field called differentiable programming, the aim of which is to express computational graphs more synthetically and modularly by resorting to actual programming languages endowed with control flow operators and higher-order combinators, such as map and fold.  In this paper, we extend the backpropagation algorithm to a paradigmatic example of such a programming language: we define a compositional program transformation from the simply-typed lambda-calculus to itself augmented with a notion of linear negation, and prove that this computes the gradient of the source program with the same efficiency as first-order backpropagation. The transformation is completely effect-free and thus provides a purely logical understanding of the dynamics of backpropagation.



\end{abstract}



\maketitle

\section{Introduction}\label{sect:intro}

In the past decade there has been a surge of interest in so-called \emph{deep learning}, a class of machine learning methods based on multi-layer neural networks. The term ``deep'' has no formal meaning, it is essentially a synonym of ``multi-layer'', which refers to the fact that the networks have, together with their input and output layers, at least one internal (or ``hidden'') layer of artificial neurons. These are the basic building blocks of neural networks: technically, they are just functions $\Real^m\to\Real$ of the form $(x_1,\ldots,x_m)\to \sigma\!\left(\sum_{i=1}^m w_i\cdot x_i\right)$, where $\sigma:\Real\to\Real$ is some \emph{activation function} and $w_1,\ldots,w_m\in\Real$ are the \emph{weights} of the neuron. A layer consists of several unconnected artificial neurons, in parallel. The simplest way of arranging layers is to connect them in cascade, every output of each layer feeding into every neuron of the next layer, forming a so-called \emph{feed-forward architecture}. Feed-forward, multi-layer neural networks are known to be universal approximators: any continuous function $f:K\to\Real$ with $K\subseteq\Real^k$ compact may be approximated to an arbitrary degree of precision by a feed-forward neural network with one hidden layer, as long as the weights are set properly~\cite{Cybenko,Hornik}. This leads us to the question of \emph{how to efficiently train a neural network}, \ie, how to find the right weights as quickly as possible.

Albeit in theory the simplest architecture suffices for all purposes, in practice more complex architectures may perform better, and different architectures may be better suited for different purposes. We thus get to another basic question of deep learning: \emph{how to select and, if necessary, modify or design network architectures adapted to a given task}. Since the quality of an architecture is also judged in terms of training efficiency, this problem is actually interlaced with the previous one.

The first question is generally answered in terms of the \emph{gradient descent} algorithm (or some variant thereof). This finds local minima of a function $G:\Real^n\to\Real$ using its gradient $\nabla G$, \ie, the vector of the partial derivatives of $G$ (Equation~\ref{eq:Grad}). The algorithm starts by choosing a point $\seq w_0\in\Real^n$. Under certain assumptions, if $\nabla G(\seq w_0)$ is close to zero then $\seq w_0$ is within a sufficiently small neighborhood of a local minimum. Otherwise, we know that $G$ decreases most sharply in the opposite direction of $\nabla G(\seq w_0)$, and so the algorithm sets $\seq w_1:=\seq w_0 - \rho\nabla G(\seq w_0)$ for a suitable step rate $\rho>0$, and repeats the procedure from $\seq w_1$. In the case of deep learning, a neural network (with one output) induces a function $h(\seq w,\seq x):\Real^{n+k}\to\Real$, where $k$ is the number of inputs and $n$ the number of weights of all the neurons in the network. By fixing $\seq w\in\Real^n$ and making $\seq x$ vary over a set of \emph{training inputs}, we may measure how much $h$ differs from the target function $f:\Real^k\to\Real$ when the weights are set to $\seq w$. This defines an error function $G:\Real^n\to\Real$, the minimization of which gives a way of finding the desired value for the weights. The training process thus becomes the iteration, over and over, of a gradient computation, starting from some initial set of weights $\seq w_0$. 

\begin{figure}[t!]
\scalebox{\scalefact}{%
\begin{tikzpicture}[
	triangle/.style = {regular polygon, regular polygon sides=3, shape border rotate=270}
]
	\znode{x1}{0}{1.75}{$x_1$}
	\znode{x2}{0}{0.25}{$x_2$}
	\znode{z1}{2}{1}{$z_1$}
	\znode{z2}{4}{1}{$z_2$}
	\znode{y}{6.25}{1}{$y$}
	\zbiedge{0.85}{1}{$\,-\,$}{x1}{x2}{z1}
	\zedge{mul}{3}{1}{$\ \cdot\ $}
	\draw[->] (z1) to[bend left] (mul);
	\draw[->] (z1) to[bend right] (mul);
	\draw[->] (mul) to (z2);
	\zunedge{5}{1}{$\sin$}{z2}{y}
\end{tikzpicture}}
\begin{center}
	\lstinline|let z$_1$ = x$_1$ - x$_2$ in let z$_2$ = z$_1$ $\Mult{}{}$ z$_1$ in $\sin$ z$_2$|
\end{center}

\caption{A computational graph with inputs $x_1,x_2$ and output $y$, and its corresponding term. Nodes are drawn as circles, hyperedges as triangles. The output $y$ does not appear in the term: it corresponds to its root.}
\label{fig:CompGraph}
\end{figure}
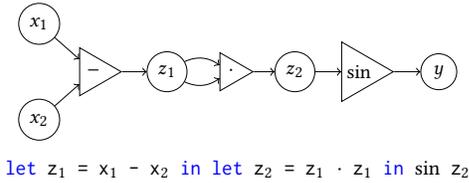

So, regardless of the architecture, efficiently training a neural network involves efficiently computing gradients. The interest of gradient descent, however, goes well beyond deep learning, into fields such as physical modeling and engineering design optimization, each with numerous applications. It is thus no wonder that a whole research field, known as \emph{automatic differentiation} (AD for short), developed around the computation of gradients and, more generally, Jacobians\footnote{The generalization of the gradient to the case of functions $\Real^n\to\Real^m$ with $m>1$.}~\cite{ADSurvey}. In AD, the setting of neural networks is generalized to \emph{computational graphs}, which are arbitrarily complex compositions of nodes computing basic functions and in which the output of a node may be shared as the input of an unbounded number of nodes. \reffig{CompGraph} gives a pictorial representation of a simple computational graph made of only three basic functions (subtraction, multiplication and sine), with inputs $x_1$ and $x_2$. Notice that the computation of $x_1-x_2$ is shared by the two factors of the multiplication. Neural networks are special cases of computational graphs.

The key idea of AD is to compute the gradient of a computational graph by accumulating in a suitable way the partial derivatives of the basic functions composing the graph. This rests on the mathematical principle known as \emph{chain rule}, giving the derivative of a composition $f\circ g$ from the derivatives of its components $f$ and $g$ (Equation~\ref{eq:ChainRule}). There are two main ``modes'' of applying this rule in AD, either \emph{forward}, propagating derivatives from inputs to outputs, or \emph{backwards}, propagating derivatives from outputs to inputs. As will be explained in \refsect{tutorial}, if $G$ is a computational graph with $n$ inputs and $m$ outputs invoking $\size G$ operations (\ie, nodes), forward mode computes the Jacobian of $G$ in $O(n\size G)$ operations, while reverse mode requires $O(m\size G)$ operations. In deep learning, as the number of layers increases, $n$ becomes astronomical (millions, or even billions) while $m=1$, hence the reverse mode is the method of choice and specializes in what is called the \emph{backpropagation algorithm} (\refsect{Backprop}), which has been a staple of deep learning for decades~\cite{LeCunBDHHHJ}. Today, AD techniques are routinely used in the industry through deep learning frameworks such as TensorFlow \cite{DBLP:conf/osdi/AbadiBCCDDDGIIK16} and PyTorch \cite{paszke2017automatic}.

The interest of the programming languages (PL) community in AD stems from the second deep learning question mentioned above, namely the design and manipulation of (complex) neural network architectures. As it turns out, these are being expressed more and more commonly in terms of actual programs, with branching, recursive calls or even higher-order primitives, like list combinators such as {\tt map} or {\tt fold}, to the point of yielding what some call a generalization of deep learning, branded \emph{differentiable programming}~\cite{LeCun:diff}. Although these programs always reduce, in the end, to computational graphs, these latter are inherently static and therefore inadequate to properly describe something which is, in reality, a dynamically-generated neural network. Similarly, traditional AD falls short of providing a fully general foundation for differentiable programming because, in order to compute the gradient of an \emph{a priori} arbitrary (higher order) program,  it forces us to reduce it to a computational graph first. In PL-theoretic terms, this amounts to fixing a reduction strategy, which cannot always be optimal in terms of efficiency. There is also a question of modularity: if gradients may only be computed by running programs, then we are implicitly rejecting the possibility of computing gradients modularly, because a minor change in the code might result in having to recompute everything from scratch, which is clearly unsatisfactory.

This paper is a contribution to the theoretical foundations of differentiable programming. We define a compositional program transformation $\revderSymbol$ (Table \ref{table:reverse})  extending the backpropagation algorithm from computational graphs to general simply typed $\lambda$-terms. Our framework is purely logical and therefore offers the possibility of importing tools from semantics, type systems and rewriting theory. The benefit is at least twofold, in the form of
\begin{enumerate}
	\item a soundness proof (Theorem~\ref{th:main theorem} and Corollary~\ref{cor:gradient}), which relies on the logical/compositional definition of $\revderSymbol$ and the semantics;
	\item a complexity analysis, which hinges on rewriting in the form of the \emph{linear substitution calculus}~\cite{Accattoli:LSC} and which guarantees that generalized backpropagation is at least as efficient as the standard algorithm on computational graphs.
\end{enumerate}
Although the soundness proof is based on a fixed strategy (reducing to a computational graph first and then computing the gradient), the confluence of our calculus guarantees us that the gradient may be computed using any other reduction strategy, thus allowing arbitrary evaluation mechanisms for executing backpropagation in a purely functional setting, without necessarily passing through computational graphs. \refsect{example} discusses the benefits of this in terms of efficiency.

On a similar note, compositionality ensures modularity: for instance, if $p=tu$, \ie, program $p$ is composed of subprograms $t$ and $u$, then $\revder p=\revder t\revder u$ and $\revder{t}$ and $\revder{u}$ may be computed independently, so that if $p$ is modified into $tu'$, the computation of $\revder{t}$ may be reused. Modularity also opens the way to parallelization, another potential avenue to efficiency.

Finally, let us stress that the transformation $\revderSymbol$ is remarkably simple: on ``pure'' $\lambda$-terms, \ie, not containing  any function symbol corresponding to the basic nodes of computational graphs, $\revderSymbol$ is the identity, modulo a change in the types. In particular, $\revderSymbol$ maps a datatype constructor/destructor ({\tt cons}, {\tt head}, {\tt tail}, etc.) or a typical higher order combinator ({\tt map}, {\tt fold}, etc.) to itself, which makes the transformation particularly easy to compute and analyze. We see this as a theoretical counterpart to the use of \emph{operator overloading} in implementations of AD.

From a more abstract perspective, all these properties (including compositionality) may be succinctly summarized in the fact that $\revderSymbol$ is a \emph{cartesian closed 2-functor} or, better, a \emph{morphism of cartesian 2-multicategories}, obtained by \emph{freely lifting to $\lambda$-terms a morphism defined on computational graphs}. However, such a viewpoint will not be developed in this paper.

\newcommand\revPurdue{\overleftarrow{\mathcal D}} 
\paragraph{Related work.} 
Our main source of inspiration is~\cite{Purdue}, where a program transformation $\revPurdue$ is proposed as a compositional extension of symbolic backpropagation to functional programs. We summarize their approach in \refsect{us}, let us concentrate on the main differences here:
\begin{enumerate}
	\item the transformation $\revPurdue$ uses references and delimited continuations, while our transformation $\revderSymbol$ is purely functional and only relies on a linear negation primitive on the ground type. Albeit \cite{Purdue} do mention that a purely functional version of their transformation may be obtained by encoding the memory inside the type of the continuation, this encoding adds a sequentialization (introduced by the order of the memory updates) which is absent in $\revderSymbol$ and which makes our transformation more amenable to parallelization.  
	\item The transformation $\revPurdue$ applies to a Turing-complete programming language, while we focus on the much more restrictive simply-typed $\lambda$-calculus. However, 
	no soundness proof for $\revPurdue$ is given, only testing on examples. This brings to light a difference in the general spirit of the two approaches: the paper \cite{Purdue} is mainly experimental, it comes with a deep learning framework (Lantern) and with a case study showing the relevance of this line of research. Our approach is complementary because it is mainly theoretical: we focus on proving the soundness and complexity bound of a non-trivial differentiable programming framework,
	``non-trivial'' in the sense that it is the simplest exhibiting the difficulty of the task at hand. To the best of our knowledge, such foundational questions have been completely neglected, even for ``toy'' languages, and our paper is the first to address them. 
\end{enumerate}

The earliest work performing AD in a functional setting is~\cite{Pearlmutter:2008}. Their motivations are broader than ours: they want to define a programming language with the ability to perform AD on its own programs. To this end, they endow Scheme with a combinator $\overleftarrow{\mathcal J}$ computing the Jacobian of its argument, and whose execution implements reverse mode AD. In order to do this, $\overleftarrow{\mathcal J}$ must reflectively access, at runtime, the program in which it is contained, and it must also be possible to apply it to itself. While this offers the possibility of computing higher-order derivatives (in the sense of derivative of the derivative, which we do not consider), it lacks a type-theoretic treatment: the combinator $\overleftarrow{\mathcal J}$ is defined in an untyped language. Although \cite{Pearlmutter:2008} do mention, for first-order code, a transformation essentially identical to our $\revderSymbol$ (using so-called \emph{backpropagators}), the observations that backpropagators may be typed linearly and that $\revderSymbol$ may be directly lifted to higher-order code 
(as we do in \reftab{reverse}) are original to our work. 
Finally, let us mention that in this case as well the correctness of $\overleftarrow{\mathcal J}$ is not illustrated by a mathematical proof of soundness, but by testing an implementation of it (Stalin$\nabla$).

At a more theoretical level, \cite{Elliott} gives a Haskell implementation of backpropagation extracted from a functor $\mathcal D$ over cartesian categories, with the benefit of  disclosing the compositional nature of the algorithm. However, Elliot's approach is still restricted to  first-order programs (\ie, computational graphs): as far as we understand, the functor $\mathcal D$ is cartesian but not cartesian closed, so the higher-order primitives ($\lambda$-abstraction and application) lack a satisfactory treatment. This is implicit in Sect.~4.4 of \cite{Elliott}, where the author states that he only uses biproduct categories: it is well-known that non-trivial cartesian closed biproduct categories do not exist.\footnote{If a category $\mathbf C$ has biproducts, then $1\cong 0$ (the terminal object is also initial). If $\mathbf C$ is cartesian closed, then its product $\times$ is a left adjoint and therefore preserves colimits (initial objects in particular), so $A\cong A\times 1\cong A\times 0\cong 0$ for every object $A$ of $\mathbf C$.}

We already mentioned TensorFlow and PyTorch. It is difficult at present to make a fair comparison with such large-scale differentiable programming frameworks since we are still focused on the conceptual level. Nevertheless, in the perspective of a future implementation, our work is interesting because it would offer a way of combining the benefits of hitherto diverging approaches: the ability to generate modular, optimizable code (TensorFlow) with the possibility of using an expressive language for dynamically-generated computational graphs (PyTorch).

\paragraph{Contents of the paper.} \refsect{tutorial} gives a (very subjective) introduction to the notions of AD used in this paper. \refsect{us} informally introduces our main contributions, which will then be detailed in the more technical \refsect{results}. \refsect{prel} formally specifies the programming language we work with, introducing a simply-typed $\lambda$-calculus  and a rewriting reduction (\reftab{reductions}) necessary to define and to evaluate our backpropagation transformation $\revderSymbol$. \refsect{example} applies $\revderSymbol$ to a couple of examples, in particular to a (very simple) recurrent neural network. \refsect{concl} concludes with some perspectives.

\section{A Crash Course in Automatic Differentiation}\label{sect:tutorial}
What follows is an introduction to automatic differentiation for those who are not familiar with the topic. It is extremely partial (the field is too vast to be summarized here) and heavily biased not just towards programming languages but towards the very subject of our work. It will hopefully facilitate the reader in understanding the context and achievements of the paper.

\subsection{What is automatic differentiation?}


Automatic differentiation (or AD) is the science of efficiently computing the derivative of (a restricted class of) programs~\cite{ADSurvey}. Such programs may be represented as directed acyclic hypergraphs, called \emph{computational graphs}, whose nodes are variables of type $\Real$ (the set of real numbers) and whose hyperedges are labelled by functions drawn from some finite set of interest (for example, in the context of neural networks, sum, product and some activation function), with the restriction that hyperedges have exactly one target node and that each node is the target of at most one hyperedge. The basic idea is that nodes that are not target of any hyperedge represent input variables, nodes which are not source of any hyperedge represent outputs, and a hyperedge
$$x_1,\ldots,x_k\stackrel{f}{\longrightarrow}y$$
represents an assignment $y:=f(x_1,\ldots,x_k)$, so that a computational graph with $n$ inputs and $m$ outputs represents a function of type $\Real^n\to\Real^m$. An example is depicted in \reffig{CompGraph}; it represents the function $(x_1,x_2)\mapsto \sin((x_1-x_2)^2)$.

In terms of programming languages, we may define computational graphs as generated by
\begin{lstlisting}
                        $F$, $G$ ::= x |  $f$(x$_1$,$\dots$,x$_k$) | let x = $G$ in $F$ | ($F$,$G$)
\end{lstlisting}
%
\noindent where \lstinline$x$ ranges over ground variables and $f$ over a set of real function symbols. The \lstinline|let| binders are necessary to represent \emph{sharing}, a fundamental feature of hypergraphs: \eg, in \reffig{CompGraph}, node $z_1$ is shared. The set of real function symbols consists of one nullary symbol for every real number plus finitely many non-nullary symbols for actual functions. We may write \lstinline|$f$($G_1$,$\dots$,$G_n$)| as syntactic sugar for \lstinline|let x$_1$ = $G_1$ in$\dots$ let x$_n$ = $G_n$ in $f$(x$_1$,$\dots$,x$_n$)|. Within this introduction (Sect.~\ref{sect:tutorial} and \ref{sect:us}) we adopt the same notation for a function and the symbol associated with it in the language, so for example $r$ may refer to both a real number and its associated numeral. Also, we use an {\tt OCaml}-like syntax in the hope that it will help the unacquainted reader parse the examples. From \refsect{prel} onward we will adopt a more succinct syntax. 

Typing is as expected: types are of the form $\tyR^k$ and every computational graph in context \lstinline{x$_1$:$\tyR$,$\ldots$,x$_n$:$\tyR$ $\vdash G:\tyR^m$} denotes a function $\sem{{G}}:\Real^n\to\Real^m$ (this will be defined formally in \refsubsect{first-order}).
Restricting for simplicity to the one-output case, we may say that, as far as we are concerned, \emph{the central question of AD is computing the gradient of $\sem{{G}}$ at any point $\seq r\in\Real^n$, as efficiently as possible}. We remind that the gradient of a differentiable function $f:\Real^n\to\Real$ is defined to be the function $\nabla f:\Real^n\to\Real^n$ such that, for all $\seq r\in\Real$,
\begin{equation}
	\label{eq:Grad}
	\nabla f(\seq r)=\left(\partial_1 f(\seq r),\ldots,\partial_n f(\seq r)\right),
\end{equation}
where by $\partial_i f$ we denote the partial derivative of $f$ with respect to its $i$-th argument. Of course, the above question makes sense only if $\sem{{G}}$ is differentiable, which is the case if every function symbol represents a differentiable function. In practice this is not always guaranteed (notoriously, modern neural networks use activation functions which are \emph{not} differentiable) but this is not actually an issue, as it will be explained momentarily.

Before delving into the main AD algorithms, let us pause a moment on the question of evaluating a computational graph. In terms of hypergraphs, we are given a computational graph ${{G}}$ with input nodes $x_1,\ldots,x_n$ together with an assignment $x_i:=\cval{r_i}$ with $r_i\in\Real$ for all $1\leq i\leq n$. The value $\sem{{G}}(r_1,\ldots,r_n)$ is found by progressively computing the assignments $w:=\cval{f(s_1,\ldots,s_m)}$ for each hyperedge $z_1,\ldots,z_m\stackrel{f}{\longrightarrow} w$ such that the values of all $z_i$ are already known. This is known as \emph{forward evaluation} and has cost $|{{G}}|$ (the number of nodes of ${{G}}$). 

In terms of programming languages, forward evaluation corresponds to a standard call-by-value strategy, with values being defined as tuples of numbers. For instance, for ${{G}}$ in \reffig{CompGraph}
\begin{center}
\begin{tabular}{c@{ $\red^\ast$ }c@{ $\red^\ast$ }c@{ $\red^\ast$ }c}
	\begin{minipage}{2,7cm}
	\begin{lstlisting}
	let x$_1$ = $5$ in 
	let x$_2$ = $2$ in 
		$G$
	\end{lstlisting}
	\end{minipage}
	&
	\begin{minipage}{3,5cm}
	\begin{lstlisting}
	let x$_1$ = $5$ in 
	let x$_2$ = $2$ in  
	let x$_1$ = $3$ in 
	let z$_2$ = z$_1$ $\Mult{}{}$ z$_1$ in 
		$\sin$ z$_2$
	\end{lstlisting}
	\end{minipage}
	&
	\begin{minipage}{2,8cm}
	\begin{lstlisting}
	let x$_1$ = $5$ in 
	let x$_2$ = $2$ in  
	let x$_1$ = $3$ in 
	let z$_2$ = $9$ in 
		$\sin$ z$_2$	
	\end{lstlisting}
	\end{minipage}
	&
	\begin{minipage}{2.7cm}
	\begin{lstlisting}
	let x$_1$ = $5$ in 
	let x$_2$ = $2$ in  
	let x$_1$ = $3$ in 
	let z$_2$ = $9$ in 
		$0.412$	
	\end{lstlisting}
	\end{minipage}
\end{tabular}
\end{center}
The operational semantics realizing the above computation will be introduced later. For now, let us mention that the value of a closed computational graph ${{G}}$ may be found in $O(\size {{G}})$ reduction steps, where $\size {{G}}$ is the size of ${{G}}$ as a term, consistently with the cost of forward evaluation.

\subsection{Forward mode AD}
\label{sect:Fwd}
The simplest AD algorithm is known as \emph{forward mode} differentiation. Suppose that we are given a computational graph (in the sense of a hypergraph) $G$ with input nodes $x_1,\ldots,x_n$ and one output node $y$, and suppose that we want to compute its $j$-th partial derivative in $\seq r=(r_1,\ldots,r_n)\in\Real^n$. The algorithm maintains a memory consisting of a set of assignments of the form $x:=(\cval s,\cder t)$, where $x$ is a node of $G$ and $\cval s,\cder t\in\Real$ (know as \cval{primal} and \cder{tangent}), and proceeds as follows:
\begin{itemize}
	\item we initialize the memory with $x_i:=(\cval{r_i},\cder 0)$ for all $0\leq i\leq n$, $i\neq j$, and $x_j:=(\cval{r_j},\cder 1)$.
	\item At each step, we look for a node $z_1,\ldots,z_k\stackrel{f}{\longrightarrow}w$ such that $z_i=(\cval{s_i},\cder{t_i})$ is in memory for all $1\leq i\leq k$ (there is at least one by acyclicity) and $w$ is unknown, and we add to memory
	\begin{equation}
		\label{eq:FwdGrad}
		w:=\left(\cval{f(\seq s)}\ ,\ \cder{\sum_{i=1}^k\Mult{\partial_i f(\seq s)}{t_i}}\right)
	\end{equation}
	where we used the abbreviation $\seq s:=s_1,\ldots,s_m$.
\end{itemize}
This procedure terminates in a number of steps equal to $\size G$ and one may show, using the chain rule for derivatives (which we will recall in a moment), that at the end the memory will contain the assignment $y:=(\cval{\sem G\!(\seq r)},\cder{\partial_j\!\sem G\!(\seq r)})$. 

Since the arity $k$ of function symbols is bounded, the cost of computing one partial derivative is $O(\size G)$. Computing the gradient requires computing \emph{all} $n$ partial derivatives, giving of a total cost of $O(n\size G)$, which is not very efficient since $n$ may be huge (typically, it is the number of weights of a deep neural network, which may well be in the millions).

For example, if $G$ is the computational graph of \reffig{CompGraph} and we start with $x_1:=(\cval 5,\cder 1),x_2:=(\cval 2,\cder 0)$, we obtain $z_1:=(\cval{x_1}-\cval{x_2},1\cdot \cder 1-1\cdot\cder 0)=(\cval 3,\cder 1)$, then $z_2:=(\cval{z_1}\cdot \cval{z_1}\ ,\ \cval{z_1}\cdot\cder 1+\cval{z_1}\cdot\cder 1)=(\cval 9,\cder 6)$ and finally $y:=(\sin(\cval{z_2}),\cos(\cval{z_2})\cdot\cder 6)=(\cval{0.412},\cder{-5.467})$, which is what we expect since $\partial_1\sem G(x_1,x_2)=\cos((x_1-x_2)^2)\cdot 2(x_1-x_2)$.

\subsection{Symbolic AD}\label{subsect:symbolic}
The AD algorithm presented above is purely numerical, but there is also a \emph{symbolic} approach. The basic idea of (forward mode) symbolic AD is to generate, starting from a computational graph $G$ with $n$ input nodes and $1$ output node, a computational graph $\Fwd G$ with $2n$ input nodes and $2$ output nodes such that forward evaluation of $\Fwd G$ corresponds to executing forward mode AD on $G$, \ie, for all $\seq r=r_1,\ldots,r_n\in\Real$, the output of $\Fwd G(\cval{r_1},\cder 0,\ldots,\cval{r_j},\cder 1,\ldots,\cval{r_n},\cder 0)$ is $(\cval{\sem G(\seq r)},\cder{\partial_j\sem G(\seq r)})$.

From the programming languages standpoint, symbolic AD is interesting because:
\begin{enumerate}
	\item it allows one to perform optimizations on $\Fwd G$ which would be inaccessible when simply running the algorithm on $G$ (a typical benefit of frameworks like {\tt TensorFlow});
	\item it opens the way to compositionality, with the advantages mentioned in the introduction;
	\item being a (compositional) program transformation rather than an algorithm, it offers a viewpoint from which AD may possibly be extended beyond computational graphs.
\end{enumerate}

Recently, \cite{Elliott} pointed out how symbolic forward mode AD may be understood in terms of compositionality, thus intrinsically addressing point (2) above. The very same fundamental remark is implicitly used also by \cite{Purdue}. Recall that the derivative of a differentiable function $f:\Real\to\Real$ is the (continuous) function $f':\Real\to\Real$ such that, for all $r\in\Real$, the map $a\mapsto f'(r)\cdot a$ is the best linear approximation of $f$ around $r$. Now, differentiable (resp.\ continuous) functions are closed under composition, so one may wonder whether the operation $(-)'$ is compositional, \ie, whether $(g\circ f)'= g'\circ f'$. This is notoriously not the case: the so-called \emph{chain rule} states that
\begin{equation}
	\label{eq:ChainRule}
	\forall r\in\Real,\qquad(g\circ f)'(r)=g'(f(r))\cdot f'(r).
\end{equation}

Nevertheless, there is a slightly more complex construction from which the derivative may be recovered and which has the good taste of being compositional. Namely, define, for $f$ as above,
$$
\begin{array}{rccc}
	\Der f: & \Real\times\Real & \longrightarrow & \Real\times\Real \\
	& (r,a) & \mapsto & (f(r),f'(r)\cdot a)).
\end{array}
$$
We obviously have $\pi_2\Der f(x,1)=f'(x)$ for all $x\in\Real$ (where $\pi_2$ is projection on the second component) and we invite the reader to check, using the chain rule itself, that $\Der(g\circ f)=\Der g\circ\Der f$.

The similarity with forward mode AD is not accidental. Indeed, as observed by \cite{Elliott}, forward mode symbolic AD may be understood as a compositional implementation of partial derivatives, along the lines illustrated above. Formally, we may consider two term calculi for computational graphs, let us call them $\cat C$ and $\cat C'$, defined just as above but based on two different sets of function symbols $\cF\subseteq\cF'$, respectively, with $\cF'$ containing at least sum and product, and equipped with partial functions
$$\partial_i:\mathcal F\longrightarrow\mathcal F'$$
for each positive integer $i$ such that, for all $f\in\mathcal F$ of arity $k$ and for all $1\leq i\leq k$, $\partial_i f$ is defined and its arity is equal to $k$. Then, we define a program transformation $\opFwd$ from $\cat C$ to $\cat C'$ which, on types, is given by
\begin{align*}
	\Fwd{\tyR} &:= \tyR\times\tyR & \Fwd{A\times B} &:= \Fwd A\times\Fwd B,
\end{align*}
and, on computational graphs,
\begin{center}
\begin{tabular}{r@{\;:=\;}lr@{\;:=\;}l}
	$\opFwd$(\lstinline|x|) & \lstinline|x|
	&
	$\opFwd$(\lstinline|let x = $G$ in $F$|) & \lstinline|let x = $\opFwd(G)$ in $\opFwd(F)$|
	\\[5pt]
	$\opFwd$(\lstinline|($F$,$G$)|) &\lstinline|($\opFwd(F)$,$\opFwd(G)$)|
	&
	$\Fwd{f(\seq x)}$ &
	\lstinline|let $\seq x$ = ($\seq z$, $\seq{a}$) in ($f(\seq z)$,$\sum_{i=1}^k$ $\Mult{\partial_i f (\seq z)}{\text{a}_i}$)|
\end{tabular}
\end{center}
In the last case, $f$ is of arity $k$ and \lstinline|let $\seq x$ = ($\seq z$, $\seq a$) in $\dots$| stands for \lstinline|let x$_1$ = (z$_1$,a$_1$) in $\dots$ let x$_k$ = (z$_k$,a$_k$) in $\dots$|. 
%
Notice how the case $\Fwd{f(\seq x)}$ is the definition of the operator $\Der$ mutatis mutandis, considering that now the arity $k$ is arbitrary. More importantly, we invite the reader to compare it to the assignment (\ref{eq:FwdGrad}) in the description of the forward mode AD algorithm: they are essentially identical. The following is therefore not so surprising:
\begin{proposition}
	\label{prop:Fwd}
	Suppose that every $f\in\cF$ of arity $k$ corresponds to a differentiable function $f:\Real^k\to\Real$ and that, for all $1\leq i\leq k$, $\partial_i f$ is the symbol corresponding to its partial derivative with respect to the $i$-th input. Then, for every computational graph $\seq x:\tyR\vdash G:\tyR$ with $n$ inputs, for all $1\leq j\leq n$ and for all $\seq r=r_1,\ldots,r_n\in\Real$, we have the call-by-value evaluation
\begin{center}
	\lstinline|let x$_1$ = ($r_1$,$0$) in$\dots$  let x$_j$ = ($r_j$,$1$) in$\dots$ let x$_n$ = ($r_n$,$0$) in $\Fwd G$| $ \mathrel{\red^\ast}$ \lstinline{($\sem G(\seq r)$,$\partial_j\sem G(\seq r)$)}
\end{center}
\end{proposition}

So we obtained what we wanted: evaluating $\Fwd{G}$ in call-by-value is the same as executing the forward mode AD algorithm on $G$. Moreover, the definition of $\opFwd$ is fully compositional.\footnote{Technically, one may consider the calculi $\cat C$ and $\cat C'$ as cartesian 2-multicategories: types are objects, computational graphs with $n$ inputs and one output are $n$-ary  multimorphisms and evaluation paths are 2-arrows. Then, $\opFwd$ is readily seen to be a morphism of cartesian 2-multicategories. We will not develop such categorical considerations in this paper, we will content ourselves with mentioning them in footnotes.} Indeed, we merely reformulated the transformation $\overrightarrow{\mathcal D_x}$ introduced in \cite{Purdue}: the computation of the proposition corresponds to that of $\overrightarrow{\mathcal D_{x_j}}(G)$. 

Symbolic AD also provides us with an alternative analysis of the complexity of the forward mode AD algorithm. Recall that the length of call-by-value evaluation of computational graphs is linear in the size, so the computation of \refprop{Fwd} takes $O(\size{\Fwd{G}})$ steps. Now, by inspecting the definition of $\opFwd$, it is immediate to see that every syntactic construct of $G$ is mapped to a term of bounded size (remember that the arity $k$ is bounded). Therefore, $\size{\Fwd{G}}=O(\size G)$, which is exactly the cost of computing one partial derivative in forward mode. In other words, the cost becomes simply the size of the output of the program transformation.

Notice how \refprop{Fwd} rests on a differentiability hypothesis. As mentioned above, this is not truly fundamental. Indeed, observe that $\Fwd{G}$ is \emph{always} defined, independently of whether the function symbols it uses are differentiable. This is because, in principle, the maps $\partial_i:\cF\to\cF'$ are arbitrary and the symbol $\partial_i f$ may have \emph{nothing} to do with the $i$-th partial derivative of the function that the symbol $f$ represents! Less unreasonably, we may consider ``mildly'' non-differentiable symbols and associate with them ``approximate'' derivatives: for example, the rectified linear unit $\mathrm{ReLU}(x)$ defined as  \lstinline{if x<$0$ then $0$ else x} may be mapped by $\partial_1$ to $\mathrm{Step}(x)$ defined as \lstinline{if x<$0$ then $0$ else $1$}, even though technically the latter is not its derivative because the former is not differentiable in $0$. Formally, $\mathrm{Step}$ is called a \emph{subderivative} of $\mathrm{ReLU}$. \refprop{Fwd} easily extends to subderivatives and, therefore, AD works just as well on non-differentiable functions, computing \emph{subgradients} instead of gradients, which is perfectly acceptable in practice. This remark applies also to backpropagation and to all the results of our paper, although for simplicity we prefer to stick to the more conventional notion of gradient, and thus work under a differentiability hypothesis.

\subsection{Reverse mode AD, or backpropagation}
\label{sect:Backprop}
A more efficient way of computing gradients in the many inputs/one output case is provided by \emph{reverse mode} automatic differentiation, from which the \emph{backpropagation} (often abbreviated as \emph{backprop}) algorithm derives. As usual, we are given a computational graph (seen as a hypergraph) $G$ with input nodes $x_1,\ldots,x_n$ and output node $y$, as well as $\seq r=r_1,\ldots,r_n\in\Real$, which is the point where we wish to compute the gradient. The backprop algorithm maintains a memory consisting of assignments of the form $x:=(\cval r,\crder \alpha)$, where $x$ is a node of $G$ and $\cval r,\crder \alpha\in\Real$ (the \cval{primal} and the \crder{adjoint}), plus a boolean flag with values ``marked/unmarked'' for each hyperedge of $G$, and proceeds thus:
\begin{description}
	\item[initialization:] the memory is initialized with $x_i:=(r_i,0)$ for all $1\leq i\leq n$, and the \emph{forward phase} starts;
	\item[forward phase:] at each step, a new assignment $z:=(\cval s,\crder 0)$ is produced, with $s$ being computed exactly as during forward evaluation, ignoring the second component of pairs in memory (\ie, $s$ is the value of node $z$); once every node of $G$ has a corresponding assignment in memory, the assignment $y:=(\cval t,\crder 0)$ is updated to $y:=(\cval t,\crder 1)$, every hyperedge is flagged as unmarked and the \emph{backward phase} begins (we actually know that $t=\sem{G}(\seq r)$, but this is unimportant);
	\item[backward phase:] at each step, we look for an unmarked hyperedge $z_1,\ldots,z_k\stackrel{f}{\longrightarrow} w$ such that all hyperedges having $w$ among their sources are marked. If no such hyperedge exists, we terminate. Otherwise, assuming that the memory contains $w:=(\cval u,\crder \alpha)$ and $z_i:=(\cval{s_i},\crder{\beta_i})$ for all $1\leq i\leq k$, we update the memory with the assignments $z_i:=\left(\cval{s_i}\ ,\ \crder{\beta_i+\partial_i f(\seq s)\cdot\alpha}\right)$ (where $\seq s:=s_1,\ldots,s_k$) for all $1\leq i\leq k$ and flag the hyperedge as marked.
\end{description}

One may prove that, when the backprop algorithm terminates, the memory contains assignments of the form $x_i:=\left(\cval{r_i},\crder{\partial_i\sem{G}(\seq r)}\right)$ for all $1\leq i\leq n$, \ie, the value of each partial derivative of $\sem G$ in $\seq r$ is computed at the corresponding input node, and thus the gradient may be obtained by just collecting such values. Let us test it on an example. Let $G$ be the computational graph of \reffig{CompGraph} and let $r_1=5,r_2=2$. The forward phase terminates with $x_1:=(\cval 5,\crder 0),x_2:=(\cval 2,\crder 0),z_1:=(\cval 3,\crder 0),z_2:=(\cval 9,\crder 0),y:=(\cval{0.412},\crder 1)$. From here, the backward phase updates $z_2:=(\cval 9,\cos(\cval{z_2})\cdot\crder 1)=(\cval 9,\crder{-0.911})$, then $z_1:=(\cval 3,\cval{z_1}\cdot\crder{-0.911}+\cval{z_1}\cdot\crder{-0.911})=(\cval 3,\crder{-5.467})$ and finally $x_1:=(\cval 5,1\cdot\crder{-5.467})=(\cval 5,\crder{-5.467})$ and $x_2:=(\cval 2,-1\cdot\crder{-5.467})=(\cval 2,\crder{5.467})$, as expected since $\partial_2\sem G=-\partial_1\sem G$.

Compared to forward mode AD, the backprop algorithm may seem a bit contrived (it is certainly harder to understand why it works) but the gain in terms of complexity is considerable: the forward phase is just a forward evaluation; and, by construction, the backward phase scans each hyperedge exactly once performing each time a constant number of operations. So both phases are linear in $\size G$, giving a total cost of $O(\size G)$, like forward mode. Except that, unlike forward mode, a single evaluation now already gives us the whole gradient, independently of the number of inputs!

\subsection{Symbolic backpropagation and the compositionality issue}
The symbolic methodology we saw for forward mode AD may be applied to the reverse mode too: given a computational graph $G$ with $n$ inputs and one output, one may produce a computational graph $\bp[] G$ with $n+1$ inputs and $1+n$ outputs such that, for all $\seq r=r_1,\ldots,r_n\in\Real$, the forward evaluation of $\bp[] G(\seq r,1)$ has output $(\sem G(\seq r),\nabla\sem G(\seq r))$. Moreover, $\size{\bp[] G}=O(\size G)$.

\begin{figure}[t!]
\scalebox{\scalefact}{%
\begin{tikzpicture}[
	triangle/.style = {regular polygon, regular polygon sides=3, shape border rotate=270}
]
	\znode[blue(pigment)]{x1}{0}{1.75}{$x_1$}
	\znode[blue(pigment)]{x2}{0}{0.25}{$x_2$}
	\znode[blue(pigment)]{z1}{2}{1}{$z_1$}
	\znode[blue(pigment)]{z2}{4}{1}{$z_2$}
	\znode[blue(pigment)]{y}{6.25}{1}{$y$}
	\zbiedge[blue(pigment)]{0.85}{1}{$\,-\,$}{x1}{x2}{z1}
	\zedge[blue(pigment)]{mul}{3}{1}{$\ \cdot\ $}
	\draw[->,color=blue(pigment)] (z1) to[bend left] (mul);
	\draw[->,color=blue(pigment)] (z1) to[bend right] (mul);
	\draw[->,color=blue(pigment)] (mul) to (z2);
	\zunedge[blue(pigment)]{5}{1}{$\sin$}{z2}{y}
	\znode[burgundy]{a}{6.25}{3.75}{$a$}
	\znode[burgundy]{b}{8}{3}{$b$}
	\znode[burgundy]{v}{6.25}{2.25}{$v$}
	\znode[burgundy]{cp}{10}{3.75}{$c'$}
	\znode[burgundy]{cpp}{10}{2.25}{$c''$}
	\znode[burgundy]{c}{12}{3}{$c$}
	\znode[burgundy]{cox1}{14}{3.75}{$x_1'$}
	\znode[burgundy]{cox2}{14}{2.25}{$x_2'$}
	\zunedge[burgundy]{5}{2.25}{$\cos$}{z2}{v}
	\zbiedge[burgundy]{7}{3}{$\ \cdot\ $}{a}{v}{b}
	\zedge[burgundy]{mul1}{9}{3.75}{$\ \cdot\ $}
	\draw[->,color=burgundy] (z1) to[out=60,in=200] (5,4) to[out=20,in=160] (mul1);
	\draw[->,color=burgundy] (b) to (mul1);
	\draw[->,color=burgundy] (mul1) to (cp);
	\zedge[burgundy]{mul2}{9}{2.25}{$\ \cdot\ $}
	\draw[->,color=burgundy] (z1) to[out=315,in=225] (mul2);
	\draw[->,color=burgundy] (b) to (mul2);
	\draw[->,color=burgundy] (mul2) to (cpp);
	\zbiedge[burgundy]{11}{3}{$+$}{cp}{cpp}{c}
	\znode[burgundy]{one}{12}{4}{$1$}
	\zedge[burgundy]{mul3}{13}{3.75}{$\ \cdot\ $}
	\draw[->,color=burgundy] (one) to (mul3);
	\draw[->,color=burgundy] (c) to (mul3);
	\draw[->,color=burgundy] (mul3) to (cox1);
	\znode[burgundy]{negone}{12}{2}{$-1$}
	\zedge[burgundy]{mul4}{13}{2.25}{$\ \cdot\ $}
	\draw[->,color=burgundy] (negone) to (mul4);
	\draw[->,color=burgundy] (c) to (mul4);
	\draw[->,color=burgundy] (mul4) to (cox2);
\end{tikzpicture}}
	\lstinline{let z$_1$=x$_1$-x$_2$ in let z$_2$=z$_1\cdot$z$_1$ in let b=($\cos$ z$_2$)$\Mult$a in let c=z$_1\Mult$b+z$_1\Mult$b in ($\sin$ z$_2$,(1$\Mult$c,-1$\Mult$c))}
	\caption{The computational graph $\bp[x_1,x_2,a] G$ where $G$ is in \reffig{CompGraph}, and its corresponding term.}
	\label{fig:Backprop}
\end{figure}
A formal definition of the $\mathbf{bp}$ transformation will be given in \refsubsect{first-order}. Let us look at an example, shown in \reffig{Backprop}. First of all, observe that $\bp[x_1,x_2,a] G$ contains a copy of $G$, marked in \cval{blue}. This corresponds to the forward phase. The nodes marked in \crder{red} correspond to the backward phase. Indeed, we invite to reader to check that the forward evaluation of $\bp[x_1,x_2,a]G$ with $x_1:=5$, $x_2:=2$ and $a:=1$ matches exactly the steps of the backprop algorithm as exemplified in \refsect{Backprop}, with node $b$ (resp.\ $c$, $x_1'$, $x_2'$) corresponding to the second component of the value of node $z_2$ (resp.\ $z_1$, $x_1$, $x_2$). (The nodes $v$, $c'$ and $c''$ are just intermediate steps in the computation of $b$ and $c$ which are implicit in the numerical description of the algorithm and are hidden in syntactic sugar).

Rather than examining the details of the definition of $\mathbf{bp}$, let us observe at once that, from the standpoint of programming languages, it suffers from a serious defect: unlike $\opFwd$, it is \emph{not} compositional.  
Indeed, in order to define \lstinline{$\mathbf{bp}$(let x = $G$ in $F$)}, we need to know and exploit the inner structure of $\bp[]F$ and $\bp[]G$, whereas from the definition of $\opFwd$ given above it is clear that no such knowledge is required in that case, \ie, $\Fwd{F}$ and $\Fwd{G}$ are treated as ``black boxes''. Our way of understanding previous work on the subject \cite{Elliott,Purdue} is that it was all about \emph{making symbolic backprop compositional}. This is the topic of the next Section.

\section{Our Approach to Compositional Backpropagation}\label{sect:us}
As mentioned above, the key to modular and efficient differentiable programming is making symbolic backprop (the $\mathbf{bp}$ transformation) compositional. We will show that this may be achieved by considering a programming language with a notion of linear negation. The goal of this section is to explain why \emph{negation} and why \emph{linear}.

Let us start with by looking at an extremely simple example, which is just the composition of two unary functions:
\begin{equation}
\label{eq:SimpleComp}
G := \text{\lstinline|let z = $\;f\;$ x in $\;g\;$ z|}
\end{equation}

As a hypergraph, $G$ has three nodes $x,y,z$, of which $x$ is the input and $y$ the output (corresponding to the root of $G$) and two edges $x\stackrel{f}{\longrightarrow}z$ and $z\stackrel{g}{\longrightarrow}y$. Note that, since $G$ has only one input, its gradient is equal to its derivative and forward and reverse mode AD have the same complexity. This is why the example is interesting: it shows the difference between the two algorithms by putting them in a context where they must perform essentially the same operations. In the sequel, we set $h:=\sem G$ and we denote by $f'$, $g'$ and $h'$ the derivatives of $f$, $g$ and $h$, respectively.

For what concerns forward mode, we invite the reader to check that
\begin{center}
	$\Fwd G =$
	\lstinline|let z = let (v,a) = x in ($f$ v, $\Mult{\text{(}f' \text{ v)}}{\text{a}}$) in let (w, b) = z in ($g$ w, $\Mult{(g'  \text{ w})}{\text{b}}$)|
\end{center}
We may simplify this a bit by observing that (renaming $w$ as $z$)

\begin{center}
	$\Fwd G\isub{\text{\lstinline{(x,a)}}}{\text{\lstinline{x}}} \red^\ast $
	\lstinline{let z = $f$ x in let b = ($f'$ x)$\Mult$a in ($g$ z, ($g'$ z)$\Mult$b)}
\end{center}

On the other hand, applying the definition of \refsubsect{first-order}, we get
\begin{center}
	$\bp[\text{\lstinline{x}},\text{\lstinline{a}}]{G} = $
	\lstinline{let z = $f$ x in let b = ($g'$ z)$\Mult$a in ($g$ z,($f'$ x)$\Mult$b)}
\end{center}
%
Note that, if we substitute $r\in\Real$ for \lstinline{x} and $1$ for \lstinline{a}, in both cases we obtain $(h(r),h'(r))$ in the same number of steps, as expected. However, the order in which the derivatives are computed relative to the order of the composition $g\circ f$ is different: it is covariant in forward mode, contravariant in reverse mode. 
This corresponds precisely to the behavior of the two algorithms:
\begin{itemize}
	\item in forward mode, we start with $x:=(\cval r,\cder 1)$, from which we infer $z:=(\cval{f(r)},\cder{f'(r)})$, from which we infer $y:=(\cval{g(f(r))},\cder{\Mult{g'(f(r))}{f'(r)}})$;
	\item in reverse mode, the forward phase leaves us with $x:=(\cval r,\crder 0)$, $z:=(\cval{f(r)},\crder 0)$, $y:=(\cval{g(f(r))},\crder 1)$, at which point the backward phase proceeds back from the output towards the input, inferring first $z:=(\cval{f(r)},\crder{g'(f(r))})$ and then $x:=(\cval r,\crder{\Mult{f'(r)}{g'(f(r))}})$.
\end{itemize}

A lesson we learn from this example, in the perspective of compositionality, is that both algorithms may be seen as mapping the subprograms $f$ and $g$, which have one input and one output, to two subprograms $\overrightarrow f$ and $\overrightarrow g$ (or $\overleftarrow f$ and $\overleftarrow g$) having two inputs and two outputs, but then assembling them rather differently:

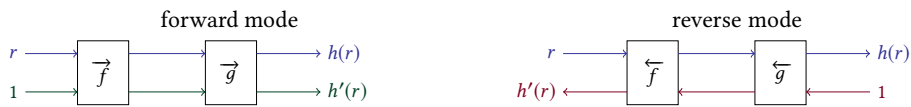
\begin{center}
\begin{tabular}{ccc}
	forward mode && reverse mode \\
	\begin{minipage}{0.4\textwidth}
		\scalebox{\scalefact}{%
		\begin{tikzpicture}
			\path (0,0) node(f) {$\overrightarrow f$};
			\draw (-0.4,-0.5) rectangle (0.4,0.5);
			\path (2,0) node(g) {$\overrightarrow g$};
			\draw (-0.4+2,-0.5) rectangle (0.4+2,0.5);
			\path (-0.4-1,0.3) node(x) [color=blue(pigment)] {$r$};
			\path (-0.4,0.3) node(a) {};
			\path (-0.4+0.8,0.3) node(b) {};
			\path (-0.4+2,0.3) node(c) {};
			\path (-0.4+2.8,0.3) node(d) {};
			\path (-0.4+3.8,0.3) node(y)[anchor=west,color=blue(pigment)] {$h(r)$};
			\path (-0.4-1,-0.3) node(x') [color=britishracinggreen] {$1$};
			\path (-0.4,-0.3) node(a') {};
			\path (-0.4+0.8,-0.3) node(b') {};
			\path (-0.4+2,-0.3) node(c') {};
			\path (-0.4+2.8,-0.3) node(d') {};
			\path (-0.4+3.8,-0.3) node(y')[anchor=west,color=britishracinggreen] {$h'(r)$};
			\draw[->,color=blue(pigment)] (x) to (a.center);
			\draw[->,color=blue(pigment)] (b.center) to (c.center);
			\draw[->,color=blue(pigment)] (d.center) to (y);
			\draw[->,color=britishracinggreen] (x') to (a'.center);
			\draw[->,color=britishracinggreen] (b'.center) to (c'.center);
			\draw[->,color=britishracinggreen] (d'.center) to (y');
		\end{tikzpicture}}
	\end{minipage}
	&\qquad&
	\begin{minipage}{0.4\textwidth}
		\scalebox{\scalefact}{%
		\begin{tikzpicture}
			\path (0,0) node(f) {$\overleftarrow f$};
			\draw (-0.4,-0.5) rectangle (0.4,0.5);
			\path (2,0) node(g) {$\overleftarrow g$};
			\draw (-0.4+2,-0.5) rectangle (0.4+2,0.5);
			\path (-0.4-1,0.3) node(x)[anchor=east,color=blue(pigment)] {$r$};
			\path (-0.4,0.3) node(a) {};
			\path (-0.4+0.8,0.3) node(b) {};
			\path (-0.4+2,0.3) node(c) {};
			\path (-0.4+2.8,0.3) node(d) {};
			\path (-0.4+3.8,0.3) node(y)[anchor=west,color=blue(pigment)] {$h(r)$};
			\path (-0.4-1,-0.3) node(x')[anchor=east,color=burgundy] {$h'(r)$};
			\path (-0.4,-0.3) node(a') {};
			\path (-0.4+0.8,-0.3) node(b') {};
			\path (-0.4+2,-0.3) node(c') {};
			\path (-0.4+2.8,-0.3) node(d') {};
			\path (-0.4+3.8,-0.3) node(y')[anchor=west,color=burgundy] {$1$};
			\draw[->,color=blue(pigment)] (x) to (a.center);
			\draw[->,color=blue(pigment)] (b.center) to (c.center);
			\draw[->,color=blue(pigment)] (d.center) to (y);
			\draw[<-,color=burgundy] (x') to (a'.center);
			\draw[<-,color=burgundy] (b'.center) to (c'.center);
			\draw[<-,color=burgundy] (d'.center) to (y');
		\end{tikzpicture}}
	\end{minipage}
\end{tabular}
\end{center}
\smallskip

The picture on the right suggested to \cite{Purdue} the idea of solving the compositionality issue via \emph{continuations}: drawing inspiration from ``There and Back Again'' \cite{DanvyGoldberg}, the blocks $\overleftarrow f$ and $\overleftarrow g$ are seen as function calls in the CPS transform of $G$, so that the forward phase takes place along the call path, while the backward phase takes place along the return path. However, in order for their approach to work, \cite{Purdue} must use references, \ie, the memory maintained by the backprop algorithm is explicitly present and is manipulated as described in \refsect{Backprop}. Moreover, since the memory is updated \emph{after} the return from each function call, they must actually resort to \emph{delimited continuations}. On the whole, although they do succeed in presenting reverse mode AD as a compositional program transformation, the work of \cite{Purdue} is closer to an (elegant!) implementation of the backprop algorithm in a functional language with references than to an abstract, purely compositional description of its dynamics.


Let us focus, instead, on the idea of contravariance. The archetypal contravariant operation is \emph{negation}. For a (real) vector space $A$, negation corresponds to the \emph{dual space} $A\multimap\Real$, which may be generalized to $A^{\bot_E}:=A\multimap E$ for an arbitrary space $E$, although in fact we will always take $E=\Real^d$ for some $d\in\Nat$. For brevity, let us keep $E$ implicit and simply write $A^\bot$. There is a canonical way, resembling CPS, of transforming a differentiable function $f:\Real\to\Real$ with derivative $f'$ into a function $\Derr f:\Real\times\Real^\bot\to\Real\times\Real^\bot$ from which the derivative of $f$ may be extracted. Namely, let, for all $x\in\Real$ and $x\dualvar\in\Real^\bot$,
\begin{equation}\label{eq:derr_def}
\Derr f(x,x\dualvar) := \left(f(x),\lambda a.x\dualvar(f'(x)\cdot a)\right),
\end{equation}
where we are using $\lambda$-notation with the standard meaning. We let the reader verify that, if we suppose $E=\Real$, then for all $x\in\Real$,
$(\pi_2\Derr f(x,I))1 = f'(x)$, 
where $\pi_2$ is projection of the second component and $I:\Real\to\Real$ is the identity function (which is obviously linear). More importantly, $\Derr$ is compositional: for all $x\in\Real$ and $x\dualvar\in\Real\multimap\Real$, we have\footnote{The equation in the second line uses both commutativity and associativity of product, but only the latter is really necessary: by replacing $f'(x)\cdot a$ with $a\cdot f'(x)$  in Eq.~\ref{eq:derr_def} one can check that commutativity is not needed. The former notation reflects that this is actually a linear application: in general, if $f:A\to B$, then $f'(x):A\multimap B$ and $a:A$.
 When $A=B=k$ with $k$ a ring (commutative or not), $k\multimap k\cong k$ and linear application becomes the product of $k$, so the notation $a\cdot f'(x)$ makes sense and backprop has in fact been applied to non-commutative rings~\cite{IsokawaKMP,PearsonBisset}. In the general case, it makes no sense to swap function and argument and it is unclear how backprop would extend.}
\begin{align*}
	\Derr g(\Derr f(x,x\dualvar)) &=
	\Derr g(f(x),\lambda a.x\dualvar(f'(x)\cdot a)) =
	(g(f(x)),\lambda b.(\lambda a.x\dualvar(f'(x)\cdot a))(g'(f(x))\cdot b)) \\
	&=
	(g(f(x)),\lambda b.x\dualvar(f'(x)\cdot(g'(f(x))\cdot b)) =
	(g(f(x)),\lambda b.x\dualvar((g'(f(x))\cdot f'(x))\cdot b)) \\
	&=
	((g\circ f)(x),\lambda b.x\dualvar((g\circ f)'(x)\cdot b)) =
	\Derr(g\circ f)(x,x\dualvar).
\end{align*}
%
This observation may be generalized to maps $f:\Real^n\to\Real$: for all $\seq x\in\Real^n$ and $\seq x\dualvar=x_1\dualvar,\ldots,x_n\dualvar\in\Real^\bot$,
$$\RevDer{f}(\seq x;\seq x\dualvar):=\left(f(\seq x),\lambda a.\sum_{i=1}^n x_i\dualvar\left(\partial_i f(\seq x)\cdot a\right)\right).$$
In the AD literature, the $x_i\dualvar$ are called \emph{backpropagators}~\cite{Pearlmutter:2008}. Obviously $\RevDer{f}:(\Real\times\Real^\bot)^n\to\Real\times\Real^\bot$ and we invite the reader to check that, if we take $E=\Real^n$, we have
$$(\pi_2\RevDer{f}(\seq x;\iota_1,\ldots,\iota_n))1=\nabla f(\seq x),$$
where, for all $1\leq i\leq n$, $\iota_i:\Real\to\Real^n$ is the injection into the $i$-th component, \ie, $\iota_i(x)=(0,\ldots,x,\ldots,0)$ with zeros everywhere except at position $i$, which is a linear function. Moreover, $\opRev$ is compositional.\footnote{Technically, functions of type $\Real^n\to\Real$ for varying $n$ form what is known as a \emph{cartesian operad}, or \emph{clone}, and $\opRev$ is a morphism of such structures. A bit more explicitly, one may compose $f:\Real^n\to\Real$ with $g:\Real^{m+1}\to\Real$ by ``plugging'' $f$ into the $i$-th coordinate of $g$, forming $g\circ^i f:\Real^{n+m}\to\Real$, for any $1\leq i\leq m+1$; the operation $\opRev$ preserves such compositions.} This leads to the definition of a compositional program transformation $\opRev$ (\reftab{reverse}) which verifies $\RevDer{\tyR}=\tyR\times\dual$ and which, applied to our example (\ref{eq:SimpleComp}), gives
\begin{center}
\begin{tabular}{cc}
$\opRev(G)\; = $
&
\begin{minipage}{10cm}
\begin{lstlisting}
let z = 
	let (v,v$\dualvar$) = x in 
	($f$ v, fun b -> v$\dualvar$ (($f'$ v) $\cdot$ b)) in 
let (w, w$\dualvar$) = z in 
($g$ w, fun a -> w$\dualvar$ ($g'$(w) $\cdot$ a))
\end{lstlisting}
\end{minipage}
\end{tabular}
\end{center}
%
where \lstinline{w$\dualvar$, v$\dualvar$ : $\dual$} so that both \lstinline{fun b -> v$\dualvar$ (($f'$ v) $\cdot$ b))}  and  \lstinline{fun a -> w$\dualvar$ (($g'$ w) $\cdot$ a)} have also type \lstinline{$\dual$}. Notice the resemblance of $\opRev(G)$ and $\opFwd(G)$: this is not an accident, both are defined in a purely compositional way (in particular, $\opRev$(\lstinline{let x = $H$ in $F$}) = \lstinline{let x = $\opRev(H)$ in $\opRev(F)$}), abiding by the ``black box'' principle) and the only non trivial case is when they are applied to function symbols. Moreover, we have

\begin{center}
$\opRev(G)$\{\lstinline{(x,x$\dualvar$)}/\lstinline{x}\}\; $\red^\ast$\;
\lstinline{let z = $f$ x in ($g$ z, fun a -> let b = ($g'$ z)$\cdot$a in x$\dualvar$(($f'$ x)$\cdot$b))}
\end{center}

\noindent which, albeit of different type, is essentially identical to $\bp[x,a]{G}$. 
More precisely, if we write \lstinline{let b = ($g'$ z)$\cdot$a in ($f'$ x)$\cdot$b} as $F$, then we have
\begin{center}
$\bp[x,a]{G} = $\lstinline|let z=$f$ x in ($g$ z,$F$)| and 
$\opRev(G)$
\{\lstinline|(x,x$\dualvar$)|/\lstinline{x}\} $\mathrel{\red^\ast}$ \lstinline|let z=$f$ x in ($g$ z, fun a -> x$\dualvar$ $F$)|
\end{center}

Let us now explain why negation must be \emph{linear}. The above example is too simple for illustrating this point, so from now on let $G$ be the computational graph of \reffig{CompGraph}. Applying the definition of \reftab{reverse} and simplifying, we obtain that $\revder{\text{\lstinline|G|}}$ is equal to:

\newcommand{\hi}[1]{\textcolor{hilight}{#1}}
\begin{center}
\begin{tabular}{c@{ $\red^\ast$ }c}
\begin{minipage}{6.8cm}
\begin{lstlisting}
let z$_2$ =
	let z$_1$ = 
		let (v$_2$, v$_2\dualvar$) = x$_2$ in		
		let (v$_1$, v$_1\dualvar$) = x$_1$ in
		(v$_1$-v$_2$, fun c -> v$_1\dualvar$($1\cdot$c)+v$_2\dualvar$($-1\cdot$c)) in 
	let (w$_1$,w$_1\dualvar$) = z$_1$ in 
	($\Mult{\text{w}_1\!}{\!\text{w}_1}$, fun b -> w$_1\dualvar$(w$_1\cdot$b)+w$_1\dualvar$(w$_1\cdot$b)) in
let (w$_2$, w$_2\dualvar$) = z$_2$ in
($\sin$ w$_2$, fun a -> w$_2\dualvar$ (($\cos$ w$_2$)$\cdot$a))
\end{lstlisting}
\end{minipage}
&
\begin{minipage}{6.1cm}
\begin{lstlisting}
let (v$_2$, v$_2\dualvar$) = x$_2$ in		
let (v$_1$, v$_1\dualvar$) = x$_1$ in
let (z$_1$,z$_1\dualvar$) = 
	(v$_1$-v$_2$, @fun c -> v$\color{hilight}_1\dualvar$($\color{hilight}1\cdot$c)+v$\color{hilight}_2\dualvar$($\color{hilight}-1\cdot$c)@) in 
let (z$_2$, z$_2\dualvar$) =  
	(z$_1\cdot$z$_1$, fun b -> @z$\color{hilight}_1\dualvar$@(z$_1\cdot$b)+@z$\color{hilight}_1\dualvar$@(z$_1\cdot$b)) in
($\sin$ z$_2$, fun a -> z$_2\dualvar$ (($\cos$ z$_2$)$\cdot$a))
\end{lstlisting}
\end{minipage}
\end{tabular}
\end{center}

There is a potential issue here, due to the presence of two occurrences of \lstinline{z$_1\dualvar$} (highlighted in \hi{brown}) which are matched against the function corresponding to the derivative of \lstinline{x$_1$-x$_2$} (also highlighted in \hi{brown}, let us denote it by $\hi F$). Notice that such a derivative is present \emph{only once} in $\bp[x_1,x_2,a]{G}$: it corresponds to the rightmost nodes of \reffig{Backprop} (more precisely, the abstracted variable \lstinline{c} corresponds to node $c$ itself, whereas \lstinline{v$_1\dualvar$} and \lstinline{v$_2\dualvar$} correspond to nodes $x_1'$ and $x_2'$, respectively). Therefore, duplicating $\hi{F}$ would be a mistake in terms of efficiency.

The key observation here is that \lstinline{z$_1\dualvar$} is of type $\dual$, \ie, it is a \emph{linear} function (and indeed, $\hi{F}$ is linear: by distributivity of product over sum, \lstinline{c} morally appears only once in its body). This means that, for all $t,u:\tyR$, we have $z_1\dualvar t + z_1\dualvar u = z_1\dualvar (t + u)$. In the $\lambda$-calculus we consider (\refsect{prel}), this becomes an evaluation step oriented from left to right, called \emph{linear factoring} \eqref{red:linear}, allowing us to evaluate $\revder{G}$ with the same efficiency as $\bp[x_1,x_2,a]{G}$. The linear factoring $ft+fu\red f(t+u)$ would be semantically unsound in general if $f:\lnot\tyR=\tyR\to\tyR$, which is why we must explicitly track the linearity of negations.

So we have a compositional transformation $\opRev$ which takes a computational graph $G$ with $n$ inputs $x_1,\ldots,x_n$ and returns a program $x_1:\tyR\times\dual,\ldots,x_n:\tyR\times\dual\vdash\revder{G}:\tyR\times\dual$ in the simply-typed $\lambda$-calculus augmented with linear negation, such that $\revder{G}$ evaluates to (essentially) $\bp[x_1,\ldots,x_n,a]G$. Actually, thanks to another nice observation of \cite{Purdue}, we can do much more: we can extend $\opRev$ for free to the \emph{whole} simply-typed $\lambda$-calculus, just letting $\revder{\text{\lstinline{fun x ->}}t}:=\text{\lstinline{fun x ->}}\revder{t}$ and $\revder{tu}:=\revder{t}\revder{u}$, and, whenever $\text{\lstinline{x}}_1:\tyR,\ldots,\text{\lstinline{x}}_n:\tyR\vdash t:\tyR$, we have that $\revder{t}$ still computes $\nabla\sem{t}$ with the same efficiency as the evaluation of $t$! Indeed, the definition of $\opRev$ immediately gives us that if $t\red^\ast u$ in $p$ steps, then $\revder{t}\red^\ast\revder{u}$ in $O(p)$ steps (point 2 of \reflemma{Commutation}).\footnote{Morally, the simply-typed $\lambda$-calculus augmented with a set $\cF$ of function symbols is the free cartesian semi-closed 2-multicategory on $\cF$ (semi-closed in the sense of \cite{Hyland}). Therefore, once $\opRev$ is defined on $\cF$, it automatically extends to a morphism of such structures. In particular, it functorially maps evaluations (which are 2-arrows) to evaluations.} But since $t$ has ground type and ground free variables, eliminating all higher-order redexes gives $t\red^\ast G$ for some computational graph $G$, hence $\revder{t}$ evaluates to (essentially) $\bp[]G$. So $\revder{t}$ computes the gradient of $\sem t=\sem G$ (remember that the semantics is invariant under evaluation) with a cost equal to $O(\size G)$ plus the cost of the evaluation $t\red^\ast G$, which is the best we can expect in general.

To conclude, we should add that in the technical development we actually use Accattoli's \emph{linear substitution calculus} \cite{Accattoli:LSC}, which is a bit more sophisticated than the standard simply-typed $\lambda$-calculus. This is due to the presence of linear negation, but it is also motivated by efficiency, which is the whole point of backpropagation. To be taken seriously, a proposal of using functional programming as the foundation of (generalized) AD ought to come with a thorough complexity analysis ensuring that we are not losing efficiency in moving from computational graphs to $\lambda$-terms. Thanks to its tight connection with abstract machines and reasonable cost models~\cite{Accattoli:2014}, 
ultimately owed to its relationship with Girard's proof nets \cite{Accattoli18} (a graphical representation of linear logic proofs which may be seen as a higher order version of computational graphs), the linear substitution calculus is an ideal compromise between the abstractness of the standard $\lambda$-calculus and the concreteness of implementations, and provides a solid basis for such an analysis.


\section{The Linear Substitution Calculus}\label{sect:prel}
\paragraph*{Terms and Types.}
Since the linear factoring rule \eqref{red:linear} is type-sensitive (as mentioned above, it is unsound in general), it is convenient to adopt a Church-style presentation of our calculus, \ie, with explicit type annotations on variables. The set of \emph{types} is generated by the following grammar:
\begin{align*}
A,B,C&::= \tyR \gsep A \times B \gsep A\to B \gsep \dual[d] && \text{(simple types)}
\end{align*}
where $\tyR$ is the ground type of real numbers. The negation $\dual[d]$ corresponds to the linear implication (in the sense of linear logic \cite{girard87tcs}) 
$\tyR\multimap \tyR^d.$ However, in order to keep the calculus as simple as possible, we avoid introducing a fully-fledged bilinear application (as for example in the bang-calculus \cite{EhrhardG16}) and opt instead for just a negation operator and dedicated typing rules. We may omit the subscript $d$ in $\dual[d]$ if clear from the context or unimportant.

An \emph{annotated variable} is either $x^{\exptype A}$ (called \emph{exponential variable}) with $A$ any type, or $x^{\lintype}$ (called \emph{linear variable}): the writing $x^{\gentype A}$ stands for one of the two annotations (in the linear case $A=\tyR$). The grammar of \emph{values} and \emph{terms} is given by mutual induction as follows, with $x^{\gentype A}$ varying over the set of annotated variables, $r$ over the set of real numbers $\mathbb R$ and $f$ over a finite set  $\mathcal F$ of function symbols over real numbers containing at least multiplication (noted in infix form $t\cdot u$):
\begin{align}\label{eq:terms}
	v &::=  x^{\gentype A}  \gsep \num r\gsep  \lambda x^{\gentype A}.t\gsep \pair{v_1,v_2} &&\text{(values)}\\
	t,u &::= v\gsep tu\gsep \pair{t,u}\gsep\unpair t {x^{\exptype A}} {y^{\exptype B}} u \gsep t\esub{x^{\gentype A}}{u} \gsep t+u 
	\gsep f(t_1,\ldots,t_k)&&\text{(terms)}
\end{align}
Since binders contain type annotations, bound variables will never be annotated in the sequel. In fact, we will entirely omit type annotations if irrelevant or clear from the context. Terms of the form $\num r$ are called \emph{numerals}. The term $t\esub{x}{u}$ (and its binary version $\unpair t {x} {y} u$) may be understood as the more familiar \lstinline|let x = $u$ in $t$| used in the previous informal sections. 


We denote by $\DualTerms[\mathcal F]$ the set of terms generated by the above grammar. We consider also the subset $\Terms[\mathcal F]$ of terms obtained by discarding linear negation and linear variables. 

We denote by $\size t$ the \emph{size} of a term $t$, \ie, the number of symbols appearing in $t$. We denote by $\fv{t}$ the set of \emph{free variables} of $t$, abstractions and explicit substitutions being the binding operators. As usual, $\alpha$-equivalent terms are treated as equal. A term $t$ is \emph{closed} if $\fv t=\emptyset$. In the following, we will use boldface metavariables to denote sequences: $\seq x$ will stand for a sequence of variables $x_1,\ldots,x_n$, $\seq t$ for a sequence of terms $t_1,\ldots,t_n$, etc. The length of the sequences will be left implicit, because either clear from the context or irrelevant. 

We use $n$-ary products $\pair{t_1,\dots,t_{n-1},t_n}$ as syntactic sugar for $\pair{t_1,\ldots\pair{t_{n-1},t_n}\dots}$, and we define \emph{Euclidean types} by $\tyR^d:= \tyR\times(\ldots\tyR \times \tyR\ldots)$. 
It will be useful to denote a bunch of sums without specifying the way these sums are associated. The notation
$$\sum_{i\in I}t_i,$$
will denote such a bunch for $I$ a finite set. In case $I$ is a singleton, the sum denotes its unique element. An empty sum of type $\tyR^d$ stands for $\pair{\num 0,\ldots,\num 0}$, which we denote by $\num{\seq 0}$. Of course this notation would denote a unique term modulo associativity and commutativity of $+$ and neutrality of $\num 0$, but we do not need to introduce these equations in the calculus.  

\begin{table}
\small
$$
\infer{\Gamma\vdash_{z} z:\tyR}{}
\qquad
\infer{\Gamma,x^{\exptype A}\vdash x:A}{}
\qquad
\infer{\Gamma\vdash_{(z)}\pair{t,u}:A\times B}{\Gamma\vdash_{(z)} t:A & \Gamma\vdash_{(z)} u:B}
\qquad
\infer{\Gamma\vdash_{(z)}\unpair{t}{x^{\exptype A}}{y^{\exptype B}}{u}:C}{\Gamma\vdash u:A\times B & \Gamma,x^{\exptype A},y^{\exptype B}\vdash_{(z)} t:C}
$$

$$
\infer{\Gamma\vdash\lambda x^{\exptype A}.t : A\to B}{\Gamma,x^{\exptype A}\vdash t:B}
\qquad
\infer{\Gamma\vdash tu:B}{\Gamma\vdash t:A\to B & \Gamma\vdash u:A}
\qquad
\infer{\Gamma\vdash\lambda z^{\lintype}.t : \dual[d]}{\Gamma\vdash_{z} t:\tyR^d}
\qquad
\infer{\Gamma\vdash_{(z)} tu:\tyR^d}{\Gamma\vdash t:\dual[d] & \Gamma\vdash_{(z)} u:\tyR}
$$

$$
\infer{\Gamma\vdash_{(z)} t\esub{x^{\exptype A}}{u}:C}{\Gamma\vdash u:A & \Gamma,x^{\exptype A}\vdash_{(z)} t:C}
\qquad
\infer{\Gamma\vdash_{(z')} t\esub{z^\lintype}{u}:\tyR^d}{\Gamma\vdash_{(z')} u:\tyR & \Gamma\vdash_{z} t:\tyR^{d}}
\qquad
\infer{\Gamma\vdash f(t_1,\ldots,t_k):\tyR}{\Gamma\vdash t_1:\tyR &\ldots& \Gamma\vdash t_k:\tyR}
\qquad
\infer{\Gamma\vdash \num r:\tyR}{r\in\mathbb R}
$$

$$
\infer{\Gamma\vdash_{(z)}t\cdot u:\tyR}{\Gamma\vdash_{(z)} t:\tyR & \Gamma\vdash u:\tyR}
\qquad
\infer{\Gamma\vdash_{(z)}t\cdot u:\tyR}{\Gamma\vdash t:\tyR & \Gamma\vdash_{(z)} u:\tyR}
\qquad
\infer{\Gamma\vdash_{z} \num{\seq 0}:\tyR^d}{}
\qquad
\infer{\Gamma\vdash_{(z)} t+u:\tyR^d}{\Gamma\vdash_{(z)} t:\tyR^d &\Gamma\vdash_{(z)} u:\tyR^d}
$$
\caption{The typing rules. In the pairing and sum rules, either all three sequents have $z$, or none does.}\label{table:types}
\end{table}
The typing rules are in \reftab{types}. The meta-variables $\Gamma, \Delta$ vary over the set of typing contexts, which are finite sequences of exponential type annotated variables without repetitions.  There are two kind of sequents: $\Gamma\vdash t:A$ and $\Gamma\vdash_z t:\tyR^d$. In this latter $d\in\mathbb N$ and $z$ is linear type annotated variable which occurs free \emph{linearly} in $t$. The typing rules define what ``occurring linearly'' means, following the standard rules of linear logic.\footnote{In this paper we focus on exactly what is required to express the backpropagation algorithm in the $\lambda$-calculus, avoiding a full linear logic typing assignment and just tracking the linearity of a single variable of type $\tyR$ in judgments typing a term with a Euclidean type $\tyR^d$. 
} The writing $\Gamma\vdash_{(z)}t:A$ stands for either $\Gamma\vdash t:A$ or $\Gamma\vdash_{z}t:A$, and in the latter case $A=\tyR^d$ for some $d$.

\paragraph*{Contexts.} We consider one-hole contexts, or simply \emph{contexts}, which are defined by the above grammar \eqref{eq:terms} restricted to the terms having exactly one occurrence of a specific variable $\{\cdot\}$, called the \emph{hole}. Meta-variables $\metaCtxt,\metaCtxtp$ will range over the set of  one-hole contexts. Given a context $\metaCtxt$ and a term $t$ we denote by $\ctxt t$ the substitution of the hole in $\metaCtxt$ by $t$ allowing the possible capture of free variables of $t$.  A particular class of contexts are the \emph{substitution contexts} which have the form of a hole followed by a stack (possibly empty) of explicit substitutions: $\{\cdot\}\esub{p_1}{t_1}\dots\esub{p_n}{t_n}$ with each $p_i$ a variable or a pair of variables. Meta-variables $\metaSubCtxt,\metaSubCtxtp$ will range over substitution contexts. If $\metaSubCtxt$ is a substitution context, we will use the notation $\subctxt{t}$ instead of $\ctxt[\metaSubCtxt]t$.

\paragraph{Rewriting rules.}
\begin{table}
\begin{subtable}{\textwidth}
	\begin{align}
		(\lambda x.t)\alpha u &\red t\esub{x}{u}\alpha
		\label{red:lambda} \\
		\unpair{s}{x}{y}{\pair{t,u}\alpha} &\red s\esub{x}{t}\esub{y}{u}\alpha
		\label{red:pair}\\
		\ctxt{x}\esub{x^{\exptype A}}{v\alpha} &\red \ctxt{v}\esub{x^{\exptype A}}{v}\alpha
		\label{red:subst}\\
		t\esub{x^{\exptype A}}{v\alpha} &\red t\alpha &&\text{if x}\not\in\fv t
		\label{red:gc}\\
		t\esub{x^{\lintype}}{v\alpha} &\red t\isub{v}{x}\alpha
		\label{red:lin}\\
		\num r\alpha+\num q\beta&\red \num{r+q}\alpha\beta
		\label{red:plus}\\
		\pair{t_1,t_2}\alpha+\pair{u_1,u_2}\beta&\red \pair{t_1+u_1,t_2+u_2}\alpha\beta
		\label{red:plusPair}\\
		\num r\alpha\cdot \num q\beta&\red \num{rq}\alpha\beta
		\label{red:times}
\end{align}
\caption{$\beta$-rules. In \eqref{red:subst}, \eqref{red:gc} and \eqref{red:lin}, $v$ is a value. In \eqref{red:subst}, $\metaCtxt$ is an arbitrary context not binding $x$. We write $\eqref{red:subst}^{\mathfrak n}$ to refer to the instance of $\eqref{red:subst}$ in which $v$ is a numeral.}
\label{table:beta}
\end{subtable}

\begin{subtable}{\textwidth}
\begin{align}
	t &\red \lambda y.ty
	\label{red:etaLambda}\\
	t\esub{x}{\pair{u,u'}} &
	\red t\esub{x}{\pair{y,y'}}\esub{\pair{y,y'}}{\pair{u,u'}}
	\label{red:etaPair}
\end{align}
\caption{$\eta$-rules. In~\eqref{red:etaLambda} $t$ has an arrow type or $\dual$. The new variables on the right-hand side of both rules are fresh.}
\label{table:eta}
\end{subtable}
\begin{subtable}{\textwidth}
\begin{align}
(x^{\dual}\alpha t)\beta+(x^{\dual}\alpha' t')\beta' &\red x^{\dual}(t+t')\alpha\beta\alpha'\beta'\label{red:linear}
\end{align}
\caption{linear factoring ($\ell$-rule for short), where we suppose that none of $\alpha,\beta,\alpha',\beta'$ binds $x$.}
\label{table:linear}
\end{subtable}
\begin{subtable}{\textwidth}
	\begin{align}
		t\esub{x}{u}\esub{y}{w} &\ \equiv\ t\esub{y}{w}\esub{x}{u} && \text{if }y\not\in\fv u\text{ and }x\not\in\fv w
		\label{eq:comm}\\
		t\esub{x}{u}\esub{y}{w} &\ \equiv\ t\esub{x}{u\esub{y}{w}} && \text{if }y\not\in\fv t
		\label{eq:esub}\\
		t\esub{x^{\exptype A}}{u} &\ \equiv\ t_{\{y/x\}}\esub{x^{\exptype A}}{u}\esub{y^{\exptype A}}{u}
		\label{eq:dup}\\
		(s\square t)\esub{x}{u} &\ \equiv\ s\esub{x}{u}\square t &&\text{ if $x\notin\fv t$}
		\label{eq:app1}\\
		(s\square t)\esub{x}{u} &\ \equiv\ s\square(t\esub{x}{u}) &&\text{ if $x\notin\fv s$}
		\label{eq:app2}
\end{align}
\caption{Structural equivalence. In \eqref{eq:dup}, $t_{\{y/x\}}$ denotes $t$ in which some (possibly none) occurrences of $x$ are renamed to a fresh variable $y$. In \eqref{eq:app1}, \eqref{eq:app2} the writing $s \square t$ stands for either $st$ or $s+t$ or $s\cdot t$ or $\pair{s,t}$.}
\label{table:structural_equivalence}
\end{subtable}
\caption{The reduction and the structural equivalence relations, where we suppose the usual convention that no free variable in one term can be captured in the other term of a relation.}
\label{table:reductions}
\end{table}

The \emph{reduction relation} is given in \reftab{reductions} divided in three sub-groups:
\begin{align*}
 \beta&:=\{\eqref{red:lambda},\eqref{red:pair},\eqref{red:subst},\eqref{red:gc},\eqref{red:lin},\eqref{red:plus},\eqref{red:plusPair},\eqref{red:times}\}&&\text{evaluation rules,}\\ 
 \eta&:=\{\eqref{red:etaLambda},\eqref{red:etaPair}\}&&\text{extensional rules,}\\
 \ell&:=\{\eqref{red:linear}\}&&\text{linear factoring.}
\end{align*}
 
In case one wants to consider numeric computations (other than sum and products), then of course one must also include suitable reduction rules associated with the function symbols: 
\begin{equation}\label{red:function}
	f(\num r_1\alpha_1,\dots,\num r_n\alpha_n)\red \num{\sem{f}(r_1,\dots,r_n)}\alpha_1\dots\alpha_n
\end{equation}
The rule~\eqref{red:lambda} transforms a $\lambda$-calculus application into an explicit substitution. The difference between the two is that the latter commutes over terms by the structural equivalence defined in \reftab{structural_equivalence}, while the former does not. Rule \eqref{red:pair} deconstructs a pair, while rule \eqref{red:subst} implements a ``micro-step'' substitution, closer to abstract machines~\cite{Accattoli:2014}. 
The special case in which $v$ is a numeral is referred to as $\eqref{red:subst}^{\mathfrak n}$. Rule \eqref{red:gc} implements garbage collection, and rule \eqref{red:lin} linear substitution. The rules \eqref{red:etaLambda} and \eqref{red:etaPair} are standard instances of $\eta$-expansion rules. They are useful in the proof of Theorem~\ref{th:main theorem}. Rule~\eqref{red:linear} has already been discussed.

Notice that $\Terms[\mathcal F]$ is a standard linear explicit substitution calculus encompassing both call-by-need and call-by-value~\cite{Accattoli:2014}. For instance, the usual by-value $\beta$-rule $(\lambda x.t)v\fred t\isub{v}{x}$ is derivable. In this respect, the reader may think of $\Terms[\mathcal F]$ as nothing but the plain simply-typed \mbox{$\lambda$-calculus}, and consider explicit substitutions as needed merely to represent computational graphs (which are obtained by restricting to the ground type $\tyR$ only). The situation is different in $\DualTerms[\mathcal F]$, in which linearity plays a key role for expressing backpropagation.

Given any $X\subseteq\beta\cup\eta\cup\ell$, we denote by $\red[X]$ the context closure of the union of the reduction relations in $X$, for any context $\metaCtxt$:
\begin{align*}
	\ctxt t\red[X]\ctxt u,&\text{ whenever } t\red[X] u.
\end{align*}
This means that we do not consider a specific evaluation strategy (call-by-value, call-by-name etc\dots), in order to be as general as possible and to allow a future analysis concerning a more efficient operational semantics.

A term $t$ is a \emph{$X$-normal form} if there is no term $u$ with $t\red[X] u$. If the set $X$ is a singleton $\{\iota\}$, we simply write $\red[\iota]$. If $X=\beta\cup\eta\cup\ell$, we write just $\fred$. If $k\in\Nat$, $\reds[k]{X}$ denotes a sequence of length $k$ of $\red[X]$, whereas $\reds{X}$ denotes a sequence arbitrary length (including null), \ie, $\reds{X}$ is the reflexive-transitive closure of $\red[X]$. Juxtaposition of labels means their union, so that, for example, $\red[\beta\eta]$ denotes the context closure of all reduction relations except \eqref{red:linear}.  

\emph{Structural equivalence} $\equiv$ is the smallest equivalence relation containing the context closure of the rules \eqref{eq:comm}--\eqref{eq:app2} in \reftab{structural_equivalence}. Morally, structural equivalence relates terms which would correspond to the same state of an abstract machine implementing the calculus, in which explicit substitutions represent pointers to memory. We refer to \cite{Accattoli:2014,AccattoliBarras} for more details. The crucial property of $\equiv$ is that it is a (strong) bisimulation with respect to $\fred$, which means in particular that it may always be postponed to the end of an evaluation (\refprop{postponement}). The following properties are standard and we give them without proof.


\begin{proposition}[Subject reduction]\label{prop:subjectReduction}
If $t\fred u$ or $t\equiv u$ and $\Gamma\vdash_{(z)} t:A$, then $\Gamma\vdash_{(z)} u:A$.
\end{proposition}
\begin{lemma}[$\equiv$ is a strong bisimulation]\label{lemma:bisimulation}
Let $\iota$ be any reduction rule and let $t'\equiv t\red[\iota] u$, then there exists $t'\red[\iota] u'$ such that $u'\equiv u$.
\end{lemma}
\begin{proposition}[postponement of $\equiv$]\label{prop:postponement}
Let $X$ be any subset of the reduction rules in \reftab{reductions} (including the variant $\eqref{red:subst}^{\mathfrak n}$) and let $t \mathrel{(\red[X]\cup\equiv)^\ast} u$ with $k$ $X$-steps, then there exists $u'$ such that $t\reds[k]{X} u'\equiv u$.
\end{proposition}
\begin{proposition}[values]\label{prop:values}
Given a closed term $t$ of type $A$, if $t$ is a $\beta$-normal form, then it is a value.
\end{proposition}
\begin{proposition}[Weak normalization]\label{prop:wn}
For every term $t$, and every set $X\subseteq\beta\ell$, there exists a $X$-normal form $u$ such that $t\reds{X} u$.
\end{proposition}

The $\beta\ell$-rewriting enjoys also strong normalization, even modulo $\equiv$, but the proof is more involved and uninteresting for our purposes, so we omit it. The strong normalization is however immediate if we restrict  the contraction rule $\eqref{red:subst}$ to numerals, a property which will be useful in the sequel.
\begin{lemma}\label{lemma:size_shrinks}
For any $t\reds{\eqref{red:subst}^{\mathfrak n}\eqref{red:gc}\eqref{red:plusPair}\eqref{red:plus}\eqref{red:times}}u$, the number of steps in the sequence is $O(\size t)$. 
\end{lemma}

\paragraph{Denotational semantics.} The cartesian closed category of sets and functions gives a denotational model for this calculus. Types are interpreted as sets, as follows:
\begin{align*}
\sem\tyR&:= \mathbb R&
\sem{A\rightarrow B}&:= \text{set of functions from }\sem{A}\text{ to }\sem{B}\\
\sem{A\times B}&:= \sem{A}\times \sem{B}&
\sem{\dual[d]}&:= \text{set of linear maps from }\mathbb R\text{ to }\mathbb R^d
\end{align*}
Notice that the restriction to linear functions in $\sem{\dual[d]}$ is such that rule ~\eqref{red:linear} is sound. The interpretation of a judgment $\Gamma \vdash t :A$ is a function $\sem t^\Gamma$ from the cartesian product $\sem \Gamma$ of the denotations of the types in $\Gamma$ to $\sem A$. The interpretation of a judgment $\Gamma \vdash_z t :\tyR^d$, is given as a function $\sem t^{\Gamma;z}$ associating with every $\seq g\in\sem\Gamma$ a linear map from $\mathbb R$ to $\mathbb R^d$. 
The definition is by induction on $t$ and completely standard (explicit substitution is functional composition). 
We omit the superscript $\Gamma$ (or $\Gamma;z$) if irrelevant. This interpretation supposes to have associated each function symbol in $\mathcal F$ with a suitable map over real numbers. By a standard reasoning, one can prove that:
\begin{proposition}[Semantic soundness]\label{prop:semanticSoundness}
Let $\Gamma \vdash_{(z)} t:A$, then $t\fred u$ or $t\equiv u$, gives $\sem{t}=\sem{u}$.
\end{proposition}

The semantic soundness gives as a by-product a light form of confluence on Euclidean types\footnote{More sophisticated notions of confluence modulo an extension of $\equiv$ hold for the reductions in \reftab{reductions}, but we avoid to discuss this point here because inessential for our purposes.}.
\begin{corollary}\label{cor:unicity_value}
Let $\vdash t:\tyR^d$ and $v$ and $v'$ be $\beta$-normal forms s.t.~$t\freds v$ and $t\freds v'$. Then $v=v'$.
\end{corollary}
\begin{proof}
It is not hard to show that the interpretation is injective on tuples of numerals, \ie, for $v,v':\tyR^d$ values, $\sem v = \sem{v'}$ implies $v=v'$. The statement then follows from Props.~\ref{prop:values} and~\ref{prop:semanticSoundness}.
\end{proof}

From now on, we will suppose that:
\begingroup
\leqnos
  \begin{equation}
	\tag{$\star$}\label{hyp:diffentiability} \text{all function symbols in $\mathcal F$ are associated by $\sem{-}$ with differentiable maps on real numbers.}
  \end{equation}
\endgroup
As mentioned at the end of Sect.~\ref{subsect:symbolic}, this hypothesis is essentially cosmetic, in that it allows us to use actual gradients and to avoid the more technical notion of subgradient (see also Sect.~\ref{sect:concl}).

Any term $x_1^{\exptype{\tyR}},\dots,x_n^{\exptype{\tyR}}\vdash t:\tyR$ is denoted by an $n$-ary map $\sem t$ over $\mathbb R$. Since differentiable functions compose, then \eqref{hyp:diffentiability} implies that $\sem t$ is also differentiable, if $t$ contains only variables of type $\tyR$. In the general case, one can use Prop.~\ref{prop:wn} in order to $\beta$-reduce $t$ into a $\beta$-normal form $u$ containing only variables of type $\tyR$. Then by Prop.~\ref{prop:semanticSoundness} $\sem t=\sem u$ and so $\sem t$ is differentiable. This justifies the following notation, given $x_1^{\exptype{\tyR}},\dots,x_n^{\exptype{\tyR}}\vdash t:\tyR$ and a vector $\seq r\in\mathbb R^n$:
\begin{equation}\label{eq:syntactalGradient}
\nabla t (\seq r) := \pair{\num{\partial_1 \sem{t}(\seq r)},\dots,\num{\partial_n \sem{t}(\seq r)}}.
\end{equation}
Section \ref{sect:results} gives an efficient way of computing $\nabla t(\seq r)$ based on the syntactic structure of $t$.

\section{The Backpropagation Transformation}\label{sect:results}

\begin{table}
\begin{subtable}{\textwidth}
	\begin{align*}
		\revder[d]{\tyR} &:= \tyR\times\dual[d] &
		\revder[d]{A\to B} &:= \revder[d]{A}\to\revder[d]{B} &
		\revder[d]{A\times B} &:= \revder[d]{A}\times\revder[d]{B}
	\end{align*}
\caption{The action of the transformation on types.}\label{table:reverseTypes}
\end{subtable}
\begin{subtable}{\textwidth}
	\begin{align*}
		\revder[d]{x^{\exptype A}} &:= x^{\exptype{\revder[d]{A}}} \\
		\revder[d]{\lambda x^{\exptype A}.t} &:= \lambda x^{\exptype{\revder[d]{A}}}.\revder[d]{t} \\
		\revder[d]{tu} &:= \revder[d]{t}\revder[d]{u} \\
		\revder[d]{\pair{t,u}} &:= \pair{\revder[d]{t},\revder[d]{u}} \\
		\revder[d]{\unpair t {x^{\exptype{A}}} {y^{\exptype{B}}} u} &:= \unpair{\revder[d]{t}} {x^{\exptype{\revder[d] A}}} {y^{\exptype{\revder[d] B}}} {\revder[d]{u}} \\
		\revder[d]{t\esub{x^{\exptype{A}}}{u}} &:= \revder[d]{t}\esub{x^{\exptype{\revder[d] A}}}{\revder[d]{u}} \\
		\revder[d]{\num r} &:= \pair{\num r,\lambda a^\tyR.\num{\seq 0}} \\
		\revder[d]{t+u} &:= \pair{x+y,\lambda a^\tyR. (x\dualvar a+y\dualvar a)}\esub{\pair{x^{\exptype{\tyR}},{x\dualvar}^{\exptype{\dual[d]}}}}{\revder[d] t}\esub{\pair{y^{\exptype{\tyR}},{y\dualvar}^{\exptype{\dual[d]}}}}{\revder[d] u}\\
		\revder[d]{f(\seq t)} &:= \unpair{\pair{f(\seq x)\ ,\ \lambda a^\tyR.\sum_{i=1}^k x_i\dualvar\left(\partial_i f(\seq x)\cdot a\right)}}{\seq x^{\exptype{\tyR}}}{\seq {x\dualvar}^{\exptype{\dual[d]}}}{\revder[d]{\seq t}}
	\end{align*}
\caption{The action of the transformation on terms.
In the definition of $\revder{f(\seq t)}$, the sequences $\seq t,\seq x,\seq x\dualvar$ have all length $k$ equal to the arity of $f$ and the notation $\esub{\pair{\seq x^{\exptype{\tyR}},\seq {x\dualvar}^{\exptype{\dual[d]}}}}{\revder{\seq t}}$ stands for $\esub{\pair{x_1^{\exptype{\tyR}},{x_1\dualvar}^{\exptype{\dual[d]}}}}{\revder{t_1}}\cdots\esub{\pair{x_k^{\exptype{\tyR}},{x_k\dualvar}^{\exptype{\dual[d]}}}}{\revder{t_k}}$. As mentioned in \refsect{us}, the variables with superscript $\ast$ correspond to \emph{backpropagators} in AD terminology~\cite{Pearlmutter:2008}.
}\label{table:reverseTerms}
\end{subtable}
\caption{The reverse gradient $\revderSymbol[d]$ relative to $\mathbb R^d$, for an arbitrary natural number $d$. }
\label{table:reverse}
\end{table}
Let us fix two sets of function symbols $\mathcal F,\mathcal F'$ such that $\mathcal F\subseteq \mathcal F'$,  together with partial functions $\partial i:\mathcal F\longrightarrow\mathcal F'$,
for each positive integer $i$ such that, for all $f\in\mathcal F$ of arity $k$ and for all $1\leq i\leq k$, $\partial_i f$ is defined and its arity is equal to $k$. In addition to the hypothesis \eqref{hyp:diffentiability}  in Sect.~\ref{sect:prel}, we also suppose:
\begingroup
\leqnos
\begin{equation}
\tag{$\star\star$}
\label{hyp:partialderiv}
 \sem{\partial_i f}:= \partial_i \sem f
\end{equation}
\endgroup

For any $d\in\mathbb N$, \reftab{reverse} defines a program transformation $\revderSymbol[d]$ from $\Terms[\mathcal F]$ to $\DualTerms[\mathcal F']$, called the \emph{reverse gradient relative to $\mathbb R^d$}. 
Given a term $\seq x^{\exptype{\tyR}}\vdash t:\tyR$, Cor.~\ref{cor:gradient} proves that $\revder[d]{t}$ computes the gradient of $t$ in at most $O(m+\size G)$ steps, where $m$ is the cost of evaluating $t$ to a computational graph $G$. Moreover, Cor.~\ref{cor:gradient} is a consequence of  Th.~\ref{th:main theorem}, proving that actually one can reduce $\revder[d]{t}$ to a term expressing the backpropagation algorithm applied to \emph{any} computational graph $G$ $\beta$-equivalent to $t$ (indeed, to a single $\lambda$-term $t$ one can associate different computational graphs, with different sizes and with different degrees of sharing).  \refsubsect{first-order} sets the framework needed to state and prove our results. Grammar \eqref{def:ground_term} defines the notion of ground term corresponding to a computational graph and Def.~\ref{def:bp} gives the computational graph associated with the backpropagation applied to ground terms. Prop.~\ref{prop:backprop_sound} formally proves the soundness of this algorithm. \refsubsect{ho} then moves to the more general case of $\lambda$-terms, giving the soundness of our reverse gradient transformation $\revderSymbol[d]$.

Let us underline that, in the last line of \reftab{reverseTerms}, the index $d$ of $\revderSymbol[d]$ is totally independent from the arity $k$ of the function symbol $f$ and from the indexes $i\leq k$ of the variables $x\dualvar_i$ tagging the sum introduced by $\revderSymbol[d]$. The fact that $d$ may be arbitrary (it plays a role only in the last remark of \refsubsect{ho}) is a crucial feature allowing its compositionality, in contrast with the definition of $\bp{G}$ which has to refer to $\seq x$ containing the free variables in $G$  (see the case $\bp{f(\seq y)}$ in Def.~\ref{def:bp}). This being said, we henceforth omit the index $d$.

\subsection{Backpropagation on Computational Graphs}\label{subsect:first-order}

We restrict $\Terms[\mathcal F]$ to terms of type $\tyR$ not containing higher types, deemed \emph{ground terms}, as follows:
\begin{align}\label{def:ground_term}
F,G&::=x^{\exptype{\tyR}}\gsep 
F\esub{x^{\exptype{\tyR}}}{G} \gsep 
f(x_1^{\exptype{\tyR}},\ldots,x_k^{\exptype{\tyR}})
\gsep \num r\gsep F+G
\end{align}
A term $f(G_1,\ldots,G_k)$ is considered as syntactic sugar for $f(x_1,\ldots,x_k)\esub{x_1}{G_1}\dots\esub{x_k}{G_k}$. 
Fig.~\ref{fig:CompGraph} and~\ref{fig:Backprop} give examples of ground terms with the associated computational graph. Notice that the type system of \reftab{types} assigns to any ground term $G$ a type judgment $x_1^{\exptype{\tyR}},\dots,x_n^{\exptype{\tyR}}\vdash G:\tyR$ for a suitable set of variables, so $\nabla G(\seq r)$ is defined by~\eqref{eq:syntactalGradient} and the hypothesis \eqref{hyp:diffentiability} and \eqref{hyp:partialderiv}.

%

We now define the transformation $\mathbf{bp}$ implementing symbolic backpropagation, as described on hypergraphs \eg\ in \cite[Sect.~3]{VanIwaarden}. We first introduce the following notation, evaluating some trivial sums on the fly: given two ground terms $F_0,F_1$,
\[
	F_0\oplus F_1:=
		\begin{cases}
		F_{i}&\text{if $F_{i-1}=\num 0$,}\\
		F_0+F_1&\text{otherwise}
		\end{cases}
\]

\begin{definition}\label{def:bp}
Let $\Gamma =x_1^{\exptype{\tyR}},\ldots,x_n^{\exptype{\tyR}}$. Given a ground term of type $\Gamma\vdash G:\tyR$ and a fresh variable $a^\tyR$, we define the term
$$\Gamma,a^{\exptype{\tyR}}\vdash\bp G:\tyR\times\tyR^n$$
by induction on $G$, as follows. The definition is based on the inductive invariant that
\[
	\bp G=\pair{G_0,\pair{G_1,\ldots,G_n}\alpha}\beta
\]
where 
$\alpha$ and $\beta$ are substitution contexts and $a$ does not appear free in $G_0$ or in $\beta$.

\begin{itemize}
	\item $\bp{x_i}:=\pair{x_i,\pair{\num 0,\ldots,a,\ldots,\num 0}}$, where $a$ appears at the $i$-th position in the tuple.
	\item $\bp{\num r}:=\pair{\num r,\pair{\num 0,\ldots,\num 0}}$.
	\item Let $f$ be of arity $k$ and $\seq y$ a subsequence of length $k$ of $\seq x$. Then,
	\[
		\bp[\seq x,a]{f(\seq y)}:=
		\pair{
			f(\seq y),
			\pair{
			\num 0,\dots,\num 0,
			\partial_1 f(\seq y)\cdot a,
			\num 0,\dots,\num 0,
			\partial_k f(\seq y)\cdot a,
			\num 0,\dots,\num 0
			}
		},
	\]
	where the non-zero terms in the tuple are at the positions determined by $\seq y$ within $\seq x$. 
	\item Let $G=F'\esub{z^{\exptype{\tyR}}}{F''}$ and suppose that (with $b^{\exptype\tyR}$ a new variable distinct from $a^{\exptype\tyR}$)
	\begin{align*}
		\bp[\seq x,z,a]{F'} &= \pair{F_0',\pair{F_1',\ldots,F_n',H}\alpha'}\beta',&
		\bp[\seq x,b]{F''} &= \pair{F_0'',\pair{F_1'',\ldots,F_n''}\alpha''}\beta''.
	\end{align*}
	Notice that for every $i\leq n$ we can suppose by renaming that the set of variables of $F_i'$ bound by $\alpha',\beta'$ is disjoint from the set of variables of $F_i''$
	 bound by $\alpha'',\beta''$. Also, notice that $F_0''$ has type $\seq x^{\exptype{\tyR}}\vdash F_0'':\tyR$.  So we can define:
	\begin{multline*}
	\bp[\seq x,a]{F'\esub{z^{\exptype{\tyR}}}{F''}}:=\\
	\begin{cases}
	\pair{F_0',\pair{F_1',\ldots,F_n'}\alpha'}\beta'\esub{z^{\exptype{\tyR}}}{F_0''}\beta''&\text{if }H= \num 0,\\
	\pair{F_0',\pair{F_1'\oplus F_1'',\ldots,F_n'\oplus F_n''}\alpha''\esub{b^{\exptype\tyR}}{H}\alpha'}\beta'\esub{z^{\exptype{\tyR}}}{F_0''}\beta''&\text{otherwise}.
	\end{cases}
	\end{multline*}

	\item Let $G=F'+ F''$ and  suppose that
	\begin{align*}
		\bp[\seq x,a]{F'} &= \pair{F_0',\pair{F_1',\ldots,F_n'}\alpha'}\beta',&
		\bp[\seq x,a]{F''} &= \pair{F_0'',\pair{F_1'',\ldots,F_n''}\alpha''}\beta'',
	\end{align*}
	then, 
	$\bp[\seq x,a]{F'+ F''}:=\pair{F_0'+ F_0'',\pair{F_1'\oplus F_1'',\ldots,F_n'\oplus F_n''}\alpha'\alpha''}\beta'\beta''$.
\end{itemize}
\end{definition}
\begin{lemma}\label{lemma:backprop_size}
Let $G$ be a ground term with $\fv G=\seq x$, then $\size{\bp[\seq x, a] G}= O(\size G)$.
\end{lemma}
%


\begin{proposition}[{Soundness of $\mathbf{bp}$}]\label{prop:backprop_sound}
Let $G$ be a ground term whose free variables are given by a sequence $\seq x$ of length $n$. Then for every $\seq r\in \mathbb R^n$, we have:
	$$
		\bp{G}\esub{a}{\num 1}\esub{\seq x}{\seq{\num r}}
		\;\reds[O(\size G)]{\eqref{red:subst}^{\mathfrak n}\eqref{red:gc}\eqref{red:plus}\eqref{red:times}}\equiv\;
		\pair{\num{\sem{G}(\seq{r})},\nabla G(\seq r)}.
	$$
\end{proposition}
\begin{proof}
Let $R=\{\eqref{red:subst}^{\mathfrak n}\eqref{red:gc}\eqref{red:plus}\eqref{red:times}\}$. By induction on $G$ we prove that, supposing $\bp{G}=\pair{G_0,\pair{G_1,\ldots,G_n}\alpha}\beta$, for every vector of real numbers $\seq r$, we have:
(i) $G_0\beta\esub{\seq x}{\seq{\num r}}\reds{R}\equiv \num{\sem{G}(\seq{r})}$; (ii)  for all $1\leq i\leq n$, for all real number $q$, $G_i\alpha\beta\esub{a}{\num q}\esub{\seq x}{\seq{\num r}}\reds{R}\equiv\num{\partial_{x_i} \sem{G}(\seq r)\cdot q}$.
From $(i)$ and $(ii)$: 
\begin{multline*}
\bp{G}\esub{a}{\num 1}\esub{\seq x}{\num r}\equiv\\
\pair{G_0\beta\esub{\seq x}{\seq{\num r}},\pair{G_1\alpha\beta\esub{a}{\num 1}\esub{\seq x}{\seq{\num r}},\ldots,G_n\alpha\beta\esub{a}{\num 1}\esub{\seq x}{\seq{\num r}}}}$ $\reds{R}\,\equiv
\pair{\num{\sem{G}(\seq{r})},\num{\nabla_{\seq r} \sem{G}}}.
\end{multline*}
Prop.~\ref{prop:postponement} allows to postpone all structural equivalences at the end. From Lem.~\ref{lemma:size_shrinks} and $\size{\bp{G}}=O(\size G)$ (Lem.~\ref{lemma:backprop_size}) the length of the reduction is $O(\size G)$. 

We give only the proof of $(ii)$ for $G=F'\esub{z}{F''}$, the other cases being similar or trivial. Using the notations from Def.~\ref{def:bp} we have that the $i$-th component of the gradient tuple of $\bp{G}$ together with the associated substitutions is equal to (we suppose $F_i'\oplus F_i''=F_i'+F_i''$, the other cases being simpler):
\begin{align*}
&(F_i'+F_i'')\alpha''\esub{b}{H}\alpha'\beta'\esub{z}{F_0''}\beta''\esub{a}{\num q}\esub{\seq x}{\seq{\num r}}\\
&
\equiv(F_i'+F_i'')\alpha''\esub{b}{H}\alpha'\beta'\beta''\esub{a}{\num q}\esub{\seq x}{\seq{\num r}}\esub{z}{F_0''\beta''\esub{\seq x}{\seq{\num r}}}&
\\
&
\reds{R}
(F_i'+F_i'')\alpha''\esub{b}{H}\alpha'\beta'\beta''\esub{a}{\num q}\esub{\seq x}{\seq{\num r}}\esub{z}{\num{\sem{F''}(\seq r)}}
&
{\text{by (i)}}
\\
&\equiv(F_i'+F_i'')\alpha''\alpha'\beta'\beta''\esub{a}{\num q}\esub{\seq x}{\seq{\num r}}\esub{z}{\num{\sem{F''}(\seq r)}}\esub{b}{H\alpha'\beta'\beta''\esub{a}{\num q}\esub{\seq x}{\seq{\num r}}\esub{z}{\num{\sem{F''}(\seq r)}}}&
\\
&\reds[O(\size{F'})]{R}
(F_i'+F_i'')\alpha''\alpha'\beta'\beta''\esub{a}{\num q}\esub{\seq x}{\seq{\num r}}\esub{z}{\num{\sem{F''}(\seq r)}}\esub{b}{\num{\partial_z \sem{F'}(\seq r)\cdot q}}
&
\text{by IH}
\\
&\equiv 
F_i'\alpha'\beta'\esub{a}{\num q}\esub{\seq x}{\seq{\num r}}\esub{z}{\num{\sem{F''}(\seq r)}} 
+
F_i''\alpha''\beta''\esub{b}{\num{\partial_z \sem{F'}(\seq r)\cdot q}}\esub{\seq x}{\seq{\num r}}&
\\
&
\reds{R}
\num{\partial_{x_i}\sem{F'}(\seq r)\cdot q}
+
\num{\partial_{x_i}\sem{F''}(\seq r)\cdot \left(\partial_z \sem{F'}(\seq r)\cdot q\right)}
&
\text{by IH}\\
&\red[\eqref{red:plus}]
\num{
	\left(\partial_{x_i}\sem{F'}(\seq r)+\partial_z \sem{F'}(\seq r)\cdot \partial_{x_i}\sem{F''}(\seq r)\right)\cdot q
}=\num{\partial_{x_i}\sem{G}(\seq r)\cdot q}
\end{align*}
Notice that in order to move to the last line we  use the associativity, commutativity and distributivity of $+$ and $\cdot$ over real numbers, not on the corresponding syntactic symbols. In the base cases (variables, functional symbols and numerals), one performs linear substitutions~\eqref{red:subst} as well as garbage collection~\eqref{red:gc}. In the case of $\bp[\seq x,a]{f(x_{i_1},\dots,x_{i_k}})$ one instance of rule~\eqref{red:times} is needed.
\end{proof}

Notice that the result of the evaluation of $\bp{G}\esub{a}{\num 1}\esub{\seq x}{\seq{\num r}}$ is independent from the chosen reduction sequence by \refcor{unicity_value}.

\subsection{Backpropagation on Higher-Order Programs}\label{subsect:ho}
Let us now consider the soundness of our transformation $\revder{t}$ applied to any $\lambda$-term of type $x_1^{\exptype{\tyR}},\dots,x_n^{\exptype{\tyR}}\vdash t:\tyR$. The proof uses two essential ingredients: first, we remark that the transformation $\revder{t}$ commutes with any reduction step (Lem.~\ref{lemma:Commutation}); second, we prove that $\revderSymbol$ applied to a computational graph $G$ encodes actually all the relevant information of $\bp{G}$ (Lem.~\ref{lemma:main}). Notice that Lem.~\ref{lemma:main} introduces fresh variables $x'_i$ ($i\leq n$) annotated $\exptype{\dual}$, tagging the different components of the gradient of $G$ in a sum.  We can then conclude with Th.~\ref{th:main theorem} and its Cor.~\ref{cor:gradient}.

\begin{lemma}
	\label{lemma:Commutation}
	Let $t$ be a term of $\Terms[\mathcal F]$. Then:
	\begin{enumerate}
		\item if $t\equiv t'$, then $\revder{t}\equiv\revder{t'}$,
		\item if $t\red[\iota]t'$, for $\iota$ any reduction step in \reftab{reductions}, then $\revder{t}\reds[O(1)]{X}\revder{t'}$, where $X=\{\iota\}$ for any $\iota$ but  $\eqref{red:plus}, \eqref{red:times}$, in these latter cases $X=\{\iota, \eqref{red:lambda}, \eqref{red:subst}, \eqref{red:gc}\}$.		
	\end{enumerate}
\end{lemma}

\begin{lemma}
	\label{lemma:main}
	Let $G$ be a ground term with $\fv G=\seq x=x_1^{\exptype{\tyR}},\dots,x_n^{\exptype{\tyR}}$ and suppose that $\bp{G}=\pair{G_0,\pair{G_1,\ldots,G_n}\alpha}\beta$. Then, there exists $J\subseteq \{1,\dots,n\}$ such that $i\not\in J$ implies $G_i=\num 0$ and such that
		$$\revder{G}\esub{\seq x^{\exptype{\revder\tyR}}}{\pair{\seq x^{\exptype{\tyR}},\lambda a^\tyR.{\seq x\dualvar}^{\exptype{\dual}} a}}\freds[O(\size G)]\equiv
		\pair{G_0,\lambda a.(\sum_{j\in J} x_j\dualvar G_{j})\alpha}\beta.$$	
\end{lemma}
\begin{proof}
By induction on $G$. We only consider $G=F\esub{z}{F'}$, the other cases being simpler.  Let $\bp[\seq x,z,a]{F} = \pair{F_0,\pair{F_1,\ldots,F_n,H}\alpha}\beta$ and $\bp[\seq x,b]{F'} = \pair{F_0',\pair{F_1',\ldots,F_n'}\alpha'}\beta'$. Let us also consider $H\neq\underline 0$ (the other case being simpler).
	The term $\revder{G}\esub{\seq x}{\pair{\seq x,\lambda a.\seq x\dualvar a}}$ is structural equivalent to:
	\begin{align*}
	&(\revder{F}\esub{\seq x}{\pair{\seq x,\lambda a.\seq x\dualvar a}})\esub{z}{(\revder{F'}\esub{\seq x}{\pair{\seq x,\lambda b.\seq x\dualvar b}})}
	&
	\\
	&
	\freds[O(\size{F'})]\equiv
	(\revder{F}\esub{\seq x}{\pair{\seq x,\lambda a.\seq x\dualvar a}})\esub{z}{\pair{F_0',\lambda b.(\sum_{j\in J'} x_j\dualvar F'_{j})\alpha'}\beta'}
	&\text{by IH}
	\\
	&\equiv(\revder{F}\esub{\seq x}{\pair{\seq x,\lambda a.\seq x\dualvar a}})\esub{z}{\pair{F_0',\lambda b.(\sum_{j\in J'} x_j\dualvar F'_{j})\alpha'}}\beta'&
	\\
	&\red[\eqref{red:etaPair}]\red[\eqref{red:pair}](\revder{F}\esub{\seq x}{\pair{\seq x,\lambda a.\seq x\dualvar a}}\esub{z}{\pair{z,z\dualvar}})\esub{z}{F_0'}\esub{z\dualvar}{\lambda b.(\sum_{j\in J'} x_j\dualvar F'_{j})\alpha'}\beta'
	\\
	&\red[\eqref{red:etaLambda}](\revder{F}\esub{\seq x}{\pair{\seq x,\lambda a.\seq x\dualvar a}}\esub{z}{\pair{z,\lambda a.z\dualvar a}})\esub{z}{F_0'}\esub{z\dualvar}{\lambda b.(\sum_{j\in J'} x_j\dualvar F'_{j})\alpha'}\beta'\\
	&\freds[O(\size F)]\equiv
	\pair{F_0,\lambda a.(\sum_{j\in J} x_j\dualvar F_{j}+z\dualvar H)\alpha}\beta\esub{z}{F_0'}\esub{z\dualvar}{\lambda b.(\sum_{j\in J'} x_j\dualvar F'_{j})\alpha'}\beta'
	&\text{by IH}
	\\
	&\red[\eqref{red:subst}]\red[\eqref{red:gc}]\pair{F_0,\lambda a.(\sum_{j\in J} x_j\dualvar F_{j}+(\lambda b.(\sum_{j\in J'} x_j\dualvar F'_{j})\alpha')H)\alpha}\beta\esub{z}{F_0'}\beta'\\
	&\red[\eqref{red:lambda}]\pair{F_0,\lambda a.(\sum_{j\in J} x_j\dualvar F_{j}+(\sum_{j\in J'} x_j\dualvar F'_{j})\alpha'\esub{b}H)\alpha}\beta\esub{z}{F_0'}\beta'\\
	&
	\reds[\#J\cap\#J']{\eqref{red:linear}}
	\pair{F_0,\lambda a.(\sum_{j\in J\cup J'} x_j\dualvar(F_{j}\oplus F'_{j}))\alpha'\esub{b}H\alpha}\beta\esub{z}{F_0'}\beta'
	\end{align*}	
\end{proof}

By composing the reductions in \reflemma{Commutation} and \reflemma{main}, we get:
\begin{theorem}\label{th:main theorem}
Let $t$ be a term of $\Terms[\mathcal F]$ of type $\tyR$ with $\fv t=\seq x=x_1^{\exptype{\tyR}},\dots,x_n^{\exptype{\tyR}}$. For any ground term $G$ such that $t\, (\red[\beta]\cup\equiv)\dualvar\, G$ in $m$ $\beta$-steps, we have, for a suitable $J\subseteq\{1,\dots,n\}$:
\[
	\revder t\esub{\seq x^{\exptype{\revder\tyR}}}{\pair{\seq x^{\exptype{\tyR}},\lambda a^\tyR.{\seq x\dualvar}^{\exptype{\dual}} a}} \freds[O(m+\size G)]\,\equiv \pair{G_0,\lambda a.(\sum_{j\in J} x_j\dualvar G_{j})\alpha}\beta
\]
\noindent where $\bp{G}=\pair{G_0,\pair{G_1,\ldots,G_n}\alpha}\beta$ and $\forall i\notin J$, $G_i=\num 0$. 
\end{theorem}

\begin{corollary}\label{cor:gradient}
Let $t$ be a term of $\Terms[\mathcal F]$ of type $\seq x^{\exptype{\seq\tyR}}\vdash t: \tyR$, with $\seq x=x_1^{\exptype{\tyR}},\dots,x_n^{\exptype{\tyR}}$ the free variables of $t$.  Let  $m$ be the number of $\beta$-steps needed to reduce $t$ to a ground term $G$.
For any vector $\seq r\in\mathbb R^n$, let $\nabla t(\seq r)=\pair{\num{g_1},\dots,\num{g_n}}$. Then, for a suitable $J\subseteq\{1,\dots,n\}$, with $i\notin J$, $g_i=0$, we have:
\[
	\unpair{z\dualvar\num 1}z{z\dualvar}{\revder t\esub{\seq x}{\pair{\seq x,\lambda a.\seq x\dualvar a}}}\esub{\seq x}{\num{\seq r}}\freds[O(m+\size G)]\,\equiv\sum_{j\in J}x_j\dualvar\num{g_j}.
\]
\end{corollary}

One can go further and obtain the tuple of numerals $\nabla t(\seq r)$ from $\sum_{j\in J}x_j\dualvar\num{g_j}$ just by: (i) considering $\revder[d] t$ with $d=n$ (here is the only point where we require a specific choice of $d$); (ii) replacing each $x\dualvar_i$ ($i\leq n$), morally of type $\tyR\multimap\tyR^n$,  with the corresponding injection 
$\lambda a.\pair{0\dots,0,a,0,\dots,0}$; (iii) adding all tuples resulting from the reductions \eqref{red:lambda},\eqref{red:subst},\eqref{red:gc}: $\sum_{j\in J}\pair{\num 0,\dots,\num 0,\num{g_j},\num 0\dots,\num 0}\freds\pair{\num g_1,\dots,\num g_n}$.
However, this last reduction is quadratic in $n$, as it performs $\#J=O(n)$ sums of vectors of size $n$ without considering that these latter are null but on one coordinate. It would then be preferable to add an \emph{ad hoc} read-back operation allowing to decode the tuple $\nabla t(\seq r)$ out of the tagged sum $\sum_{j\in J}x_j\dualvar\num{g_j}$ more parsimoniously. 

Notice that the proof of Cor.~\ref{cor:gradient} considers a specific reduction sequence for getting $\nabla t (\seq r)$. However, Cor.~\ref{cor:unicity_value} guarantees that the result is independent from the chosen sequence.\footnote{This is in fact the case if one admits the injections $\lambda a.\pair{0\dots,0,a,0,\dots,0}$ of type $\dual$, as discussed in the previous note.} This suggests considering more efficient strategies than that of Cor.~\ref{cor:gradient}, or even more modular by decomposing the computation along the different components of $t$. We discuss this last point in the next Section.

\section{An Example: Recurrent Neural Networks}\label{sect:example}


\subsection{Derivative of a dynamically generated polynomial}
\newcommand{\tyN}{\mathsf{Nat}}
Let $\tyN:=(\tyR\to\tyR)\to\tyR\to\tyR$ be the type of Church natural numbers and consider
$$G:=w\cdot y+x\qquad
\qquad\qquad t:= \lambda n^{\oc\tyN}.\lambda x.n(\lambda y.G)x.$$
We have $x^{\oc\tyR},y^{\oc\tyR},w^{\oc\tyR}\vdash G:\tyR$ and $w^{\oc\tyR}\vdash t:\tyN\to\tyR\to\tyR$. If $\num n$ encodes $n\in\Nat$, the term $t\,\num n$ dynamically generates the function $\lambda x.f_n(w,x)$ with $f_n(w,x):=(w^n+w^{n-1}+\cdots+w+1)\cdot x$ (we take some liberty in simplifying arithmetic expressions for the sake of readability). Notice that the free variable $x$ of $\lambda y.G$, a term to be iterated at will, is captured later in $t$. We have
$$\revder{\lambda y.G}\ \freds\lambda\!\!\pair{y,y\dualvar}\!.\pair{G,\lambda a.w\dualvar(y\cdot a)+y\dualvar(w\cdot a)+x\dualvar a}=:G',$$
where we used the shorthand $\lambda\!\!\pair{y,y\dualvar}\!.u:=\lambda p.\unpair{u}{y}{y\dualvar}{p}$. Hence, when \eg\ $n=2$,
\begin{align*}
	\revder{t\,\num{2}} &= (\lambda n.\lambda x.n\revder{\lambda y.G}x)\,\num 2\ \freds (\lambda n.\lambda x.n\,G'\,x)\,\num 2\ \fred\ \lambda x.\num 2\,G'\,x\ \fred\ \lambda x.G'(G'x) \\
	&\fred \lambda\!\!\pair{x,x\dualvar}\!.G'(G'\pair{x,x\dualvar})\ \freds\ \lambda\!\!\pair{x,x\dualvar}\!.G'\pair{wx+x,H},
\end{align*}
where $H:=\lambda a.w\dualvar(x\cdot a)+x\dualvar(w\cdot a)+x\dualvar a$. The computation continues with
\begin{align*}
	&\freds\ \lambda\!\!\pair{x,x\dualvar}\!.\pair{(w^2+w+1)x,\lambda a.w\dualvar((wx+x)\cdot a)+H(w\cdot a)+x\dualvar a} \\
	&\freds\ \lambda\!\!\pair{x,x\dualvar}\!.\pair{(w^2+w+1)x,\lambda a.w\dualvar((wx+x)\cdot a)+w\dualvar(wx\cdot a)+x\dualvar(w^2\cdot a)+x\dualvar(w\cdot a)+x\dualvar a} \\
	&\freds\ \lambda\!\!\pair{x,x\dualvar}\!.\pair{(w^2+w+1)x,\lambda a.w\dualvar((2wx+x)\cdot a)+x\dualvar((w^2+w+1)\cdot a)}
\end{align*}
We see that the argument of $w\dualvar$ (resp.\ $x\dualvar$) is the derivative of $f_2$ with respect to $w$ (resp.\ $x$). This shows how locally free variables are handled correctly. Also observe that $G'$, which is trivial here but could in principle be a very complex function, may be pre-computed independently of $n$.

\subsection{Recurrent neural networks}

Recurrent neural networks (RNNs) are meant to process inputs that are arbitrary sequences of vectors in $\Real^d$. They are heavily used for natural language processing: sequences could for instance stand for a word or a sentence, where each vector is a representation of a letter. Such networks iterate over the input to recursively build a desired output. One example is \emph{encoding recurrent neural networks}, which are used to encode a sequence of vectors $x_i\in \Real^d$ into an output vector $h \in \Real^m$. For instance, in \emph{sentiment analysis} \cite{dos2014deep,glorot2011domain,severyn2015twitter}, it is used to predict if the expressed opinion in a sentence is positive, negative or neutral.

Given a sequence $x_1, ..., x_n\in \Real^d$, the encoding RNN produces intermediate outputs $h_1, ..., h_n \in \Real^m$ recursively, by applying a single layer $L : \Real^d \times \Real^m \rightarrow \Real^m$ to both the last output value $h_i$ and the next input vector $x_{i+1}$:
\begin{align*}
h_{i+1} & = L(x_{i+1}, h_i) & L(x, h)& = \sigma (E\cdot x + R\cdot h)
\end{align*}
where $\sigma$ is the sigmoid activation function, $E$ is a $d\times m$-matrix
called the \emph{token embedding matrix}, and $R$ is a $m\times m$ matrix whose coefficients are called 
\emph{the recurrent weights}. $h_0$ can be initialized to $\seq{0}$. In practice, a RNN ends with a loss function that we are looking to minimize (and which models the problem we want to solve). To keep the example simple, we will not consider it. For the same reason, we suppose that $d = m = 1$, the general encoding being a direct generalization. 

\paragraph*{Lists}

We represent lists in our language using the Church encoding, which defines a list by its right fold function. Given $A$ and $X$ be types, the type of lists $\tyList{A}{X}$ is defined as follows:
$$\tyList{A}{X} := (A \rightarrow X \rightarrow X) \rightarrow (X \rightarrow X)$$
Given $a_1 : A, \dots, a_n: A$, we define the list $[a_1; ...; a_n]_X$ of type $\tyList{A}{X}$ as:
$$[a_1; ... ; a_n]_X := \lambda f^{\exptype{(A\rightarrow X\rightarrow X)}}.\lambda x^{\exptype{X}}.f a_1 (f a_2 \dots (f a_n x)\dots)$$
In what follows we will omit the $X$ annotation if it is clear from context or if it does not matter. 
Finally, it can be noted that the reverse gradient of the list datatype is given by:
\begin{align*}
\revder{\tyList{A}{X}} = \tyList{\revder{A}}{\revder{X}} & &
\revder{[a_1; ... ; a_n]_X} = [\revder{a_1}, \dots, \revder{a_n}]_{\revder{X}}
\end{align*}

\paragraph*{Encoding a RNN}
Let $\sigma$ be the function symbol corresponding to the sigmoid function $\Real \rightarrow \Real$. In dimension $1$, the token embedding matrix and the recurrent matrix may be represented as two terms $\lambda x^{\exptype{\tyR}}.(\epsilon\cdot x)$ and $\lambda h^{\exptype{\tyR}}.(\rho\cdot h)$ respectively, where $\epsilon^{\exptype{\tyR}}$ and $\rho^{\exptype{\tyR}}$ are two variables. The recurring layer $L$ and the recurring network $N$ is then defined and typed as follows:
\begin{align*}
  L &:=  \lambda x^{\exptype{\tyR}}.\lambda h^{\exptype{\tyR}}.\sigma(\epsilon\cdot x + \rho\cdot h)& 
  \seq{\epsilon}^{\exptype{\tyR}}, \seq{\rho}^{\exptype{ \tyR}} & \vdash L : \tyR \rightarrow \tyR \rightarrow \tyR\\
  N &:= \lambda l^{\exptype{List(\tyR, \tyR)}}.l L \num{0} 
  &
  \seq{\epsilon}^{\exptype{\tyR}}, \seq{\rho}^{\exptype{\tyR}} &\vdash N : \tyList{\tyR}{\tyR} \rightarrow \tyR
\end{align*}
The free variables $\epsilon$ and $\rho$ are the parameters that we wish to learn via gradient descent.

\paragraph*{Backpropagation}
For RNNs, the gradient is usually computed using a technique called \emph{backpropagation through time} \cite{McClellandetal,Pearlmutter}. The method consists in unfolding the RNN (by applying it to an input sequence) and then applying the usual backpropagation over plain vanilla feedforward neural networks. We will now apply our reverse gradient transformation to an example of RNN, and show that backpropagation through time is naturally implemented by the reduction strategy of Theorem~\ref{th:main theorem}.

Given $l = [\num{a_1}; \dots; \num{a_n}]_\tyR$, we compute the gradient of $N l$ with respect to $\epsilon$ and $\rho$. The following proposition exposes the recursive equations expressing the gradient computed thanks to our transformation. They are similar to the equations of backpropagation through time.

\begin{proposition}\label{prop:example-grad}
  We have $\nabla (Nl)(e, r) = \pair{g^n_{\epsilon}, g^n_{\rho}}$ where $g^n_\epsilon$ and $g^n_\rho$ are given by the following recurrent equations:
  \begin{align*}
    g^0_{\epsilon} & = \num{0} & g^{i+1}_{\epsilon} & = \sigma_{i+1}'\cdot(\num{a_{i+1}} + \num{r}\cdot g^{i}_{\epsilon})& \sigma'_{i+1} & :=  \partial_1 \sigma(\num{e}\cdot\num{a_{i+1}} + \num{r}\cdot u_i) \\
    g^0_{\rho} &= \num{0} & g^{i+1}_{\rho} &=  \sigma_{i+1}'\cdot(u_{i+1} + \num{r}\cdot g^{i}_{\rho}) & & \\
    u_0 &= \num{0} & u_{i+1} &= \sigma(\num{e}\cdot\num{a_{i+1}} + \num{r}\cdot u_{i}) & &
  \end{align*}
\end{proposition}
\begin{proof}
By Corollary \ref{cor:gradient}, this amounts to computing $D = \epsilon^* g_\epsilon + \rho^* g_\rho$ such that $$
\unpair{z^*\num 1}{z}{z^*}{\revder{Nl}\esub{\epsilon}{\pair{\epsilon,\lambda a.\epsilon^*a}}\esub{\rho}{\pair{\rho,\lambda a.\rho^*a}}}\esub{\epsilon}{\num{e}}\esub{\rho}{\num{r}} \reds{} D
$$
This is done by induction over the size of the list $l$.
\end{proof}

The strategy from Theorem \ref{th:main theorem} first computes $\revder{Nl}$, then reduces it to $\revder{G}$ such that $Nl \reds{} G$, with $G$ a ground term, \ie, only containing subterms of ground type. Here:
$$G = F(\num{a_1}, F(\num{a_2}, \dots, F(a_n, \num{0}) \dots))$$
where 
$L = \lambda x.\lambda h.F$ with $F = \epsilon\cdot x + \rho\cdot h$. Notice that $G$ is exactly the unfolding of the RNN, which should convince the reader that this strategy implements the backpropagation through time.

We now turn our attention to the effectiveness of the reduction strategy. During the reduction of $Nl$ to $G$, the $\lambda$-abstractions of $L$ must be eliminated. Each of these $\lambda$ is applied exactly to $n$ different arguments, requiring $L$ to be duplicated $n$ times. In other words, we satisfy the following reduction:
$$Nl \reds{} L\num{a_1}(L\num{a_2} \dots (L\num{a_n}\num{0})\dots) \reds{} G$$
By point 2 of \reflemma{Commutation}, the exact same reasoning may be applied to $\revder{Nl}$, in which $\revder{L}$ must be duplicated $n$ times:
$$\revder{Nl} \reds{} \revder{L}\revder{\num{a_1}}(\revder{L}\revder{\num{a_2}} \dots (\revder{L}\revder{\num{a_n}}\revder{\num{0}})\dots) \reds{} \revder{G}$$

A simple observation shows that $\revder{L}$ is not in $\beta$-normal form:
$$\revder{L} = \lambda x.\lambda h.\unpair{\pair{\sigma(\mu), \lambda a.\mu^* (\partial_1 \sigma(\mu)\cdot a)}}{\mu}{\mu^*}{\revder{\epsilon\cdot x+\rho\cdot h}}$$
$\revder{\epsilon\cdot x+\rho\cdot h}$ being a pair, we have $\revder{L} \reds[k]{\beta} u$ for some $\beta$-normal $M$. Moreover, it can be shown that when the dimensions $d$ and $m$ are $> 1$, $k = O((d+m)\times m)$. In our case, we see that each component of the sum in $\revder{\epsilon\cdot x + \rho\cdot h}$ will be responsible for at least one substitution. Finally, since $\revder{L}$ is duplicated $n$ times in the original backpropagation through time strategy, by using the strategy where $\revder{L}$ is reduced to $M$ before being substituted, we obtain a gain of $O( (d+m) m n)$ reduction steps. This is significant as the number of learning parameters (here $(d+m)m$) can be very large in typical neural networks. For instance, a recent model called GPT-2 \cite{radford2019language} has 1.5 billion parameters.

\section{Conclusion and Perspectives}\label{sect:concl}


\paragraph*{On expressiveness.} The simply-typed $\lambda$-calculus is certainly too restrictive as a programming language, but it does present the main obstacle in defining higher-order backpropagation, namely the native use of higher-order without necessarily going through computational graphs. It also has a minimum of expressiveness for interesting examples to exist, such as inductive types with very basic operations on them ({\tt map}, {\tt fold}), as shown in \refsect{example}. In fact, our results apply seamlessly to any \emph{total} language with set-theoretic semantics (\eg\ G\"odel's System T with arbitrary inductive types), with no need of additional technical ideas. This is already quite broad: no state-of-the-art deep learning architecture we are aware of (convolutional NNs, RNNs, Tree-RNNs, attention-based NNs, transformers\ldots) requires going beyond System T in order to be expressed.

We wish to stress that our current proof \emph{does} support, unchanged, non-differentiable functions like $\mathrm{ReLU}$. As mentioned at the end of Sect.~\ref{subsect:symbolic}, the differentiability hypothesis (\ref{hyp:diffentiability}) (end of \refsect{prel}) is essentially cosmetic, in that it allows us to use actual gradients and to avoid the more technical notion of subgradient. In the absence of (\ref{hyp:diffentiability}), \refcor{gradient} holds for all vectors in the domain of definition of Eq.~\ref{eq:syntactalGradient}, \ie, wherever the gradient makes sense.

Still in relation with partiality, we argue that our framework can accommodate full recursion (like \cite{Purdue}), by adding the definition $\revder{\mathsf Y_A}:=\mathsf Y_{\revder A}$ to \reftab{reverse}, where $\mathsf Y_A:(A\to A)\to A$ is the fixpoint operator. As for denotational semantics, one has to consider the cartesian closed category of pointed complete partial order sets and Scott-continuous functions. The type $\tyR$ is interpreted as the flat domain having a bottom element (representing non-termination) and all real numbers as maximal elements. \refcor{gradient} then holds for all vectors on which the program converges and is differentiable. However, this still falls short of being a full programming language like PCF because it lacks conditionals (on $\tyR$). The issues related to the presence of branchings are mentioned in~\cite{Plotkin}, where a proposal is suggested for first-order languages.

%

\paragraph*{On complexity.} 
In computational graphs, complexity analysis assumes that computing the value of a single node from its local inputs has unitary cost. In terms of low-level complexity models (\eg\ random access machines), this rests on the assumption that \emph{searching for the next node to evaluate is a constant-time operation}, which is fair because we may suppose the nodes to be linearly ordered compatibly with the dependencies of the graph.

Our analysis shows that, as soon as the programming language is suitably fine-grained, it is consistent to attribute a unitary cost to a single evaluation step (\ie, a rewriting rule from \reftab{reductions}). However, in general programming languages, it is unclear whether the search for the next redex has constant cost.  Albeit recent work \cite[Corollary 9.2]{AccattoliBarras} shows that this is the case in call-by-name evaluation of closed programs, a detailed analysis for our language is currently missing and we defer it to future work. 

\paragraph*{On higher-type derivatives.}
A major endeavor relating functional primitives with differential operators is differential linear logic \cite{Ehrhard18} and its associated differential \mbox{$\lambda$-calculus} \cite{EhrhardR03}. Indeed, the initial motivation of our work was to express the backpropagation algorithm in the differential \mbox{$\lambda$-calculus}. Although this is possible, we realized that it required discarding the most important programming primitive of the differential $\lambda$-calculus, namely the derivative operator $\mathrm D$ on higher-order types. This is probably a good place to point out a crucial feature of our program transformation $\revderSymbol$. Take a term $\lambda x.t(ux)$ of type $\tyR\rightarrow\tyR$ and such that $\vdash t:A\rightarrow \tyR$ and $\vdash u:\tyR\rightarrow A$. We have two different ways of computing the derivative of $t(ux)$ with respect to $x^\tyR$: either by our transformation $\revder{\lambda x.t(ux)}$ or by Ehrhard's derivative operator $\mathrm D (\lambda x.t(ux))$. Both ways are purely functional, but our operator decomposes as 
$\revder{\lambda x.t(ux)}=_{\beta} \lambda \pair{x, a}.\revder{t}(\revder{u}\pair{x, a})$, while Ehrhard's follows the chain rule (see Sect.~\ref{subsect:symbolic}), giving 
$\mathrm D (\lambda x.t(ux)) =_\beta \lambda\pair{x, a}. \mathrm D(t) \pair{ux, \mathrm D(u) \pair{x,a}}$.\footnote{Here we use some syntactic sugar in order to avoid notational bureaucracy. For example,  $\lambda \pair{x,a}\dots$ stands for $\lambda y.\dots\esub{\pair{x,a}}{y}$ in our language. Also, the differential $\lambda$-calculus of \cite{EhrhardR03} does not have pairs and so $\lambda\pair{x, a}$ is curried into $\lambda x.\lambda a$.} It is clear that this latter expression is inefficient with respect to backpropagation because of the duplication of the term $u$ (see also Sect.~\ref{sect:tutorial}).

So our $\revderSymbol$ and the operator $\mathrm D$ of the differential \mbox{$\lambda$-calculus} are extensionally different.  This is immediately seen on ``pure'' $\lambda$-terms (with no function symbols): when $A$ is higher order, there are pure terms $t:A\to\tyR$ with non-trivial derivative, which is \emph{always} computed by $\mathrm D(t)$, whereas $\revder t=t$ (modulo a type change), which cannot compute the derivative.\footnote{This does not contradict \refcor{gradient} because, when $A$ is of order zero (\ie, $A=\tyR^n$), any pure $t$ must be a projection and the above equality is correct.} 
Indeed, observe that our soundness statements (Theorem~\ref{th:main theorem}, Corollary~\ref{cor:gradient}) only hold for terms of ground type. At present, we seem to have no use for the extra generality provided by the differential $\lambda$-calculus, but it is a natural and intriguing question to understand whether the ability to compute derivatives at higher types (rather than just over $\tyR$) can play a role in developing new machine learning models.

\begin{acks}                            

 We would like to thank T.~Ehrhard, C.~Fouquer\'e, M.~Gaboardi, B.A.~Pearlmutter, Y.~Regis-Gianas, J.M.~Siskind, C.~Tasson and the anonymous reviewers for useful comments and discussions.

\end{acks}

\bibliography{Biblio.bib}

\newpage
\appendix
\section{Appendix}
\subsection{Proofs of Section~\ref{sect:prel}}

\noindent\textsc{Proposition \ref{prop:subjectReduction}}
{\it If $t\fred u$ or $t\equiv u$ and $\Gamma\vdash t:A$, then $\Gamma\vdash u:A$.}

\begin{proof}[Proof (Sketch)]
Let $\ctxt t\fred\ctxt u$ (resp.\ $\ctxt t\equiv\ctxt u$), where $t\red u $ (resp.\ $t\equiv u$) is a reduction step (resp.\ a structural equivalence) in \reftab{reductions}. The proof is a standard induction on the evaluation context $\metaCtxt$. The base case is by inspection of all cases in \reftab{reductions}, using the following properties:
\begin{itemize}
\item (substitution) if $\Gamma,x^A\vdash t: B$ and $\Gamma\vdash u:A$ then $\Gamma\vdash t\isub{u}{x}: B$;
\item (splitting) if $\Gamma,x^A\vdash t: B$ then for any term $t_{\{y/x\}}$ obtained from $t$ by renaming some (possibly none) occurrences of $x$ to a fresh variable $y$, we have $\Gamma, x^A, y^A\vdash t_{\{y/x\}}:B$.
\end{itemize}
The two properties are achieved by induction on $t$. Notice that the splitting property implies the so-called ``weakening'' lemma:  $\Gamma\vdash t: B$ implies  $\Gamma,x^A\vdash t: B$ for a fresh variable $x$. 
\end{proof}

\noindent\textsc{Lemma \ref{lemma:bisimulation}}
{\it Let $\iota$ be any reduction rule and let $t'\equiv t\red[\iota] u$, then there exists $t'\red[\iota] u'$ such that $u'\equiv u$.}
\begin{proof}[Proof (Sketch)]
The proof goes along the same lines as~\cite{Accattoli:2014}. We denote by $\equiv_1$ the symmetric and context closure of the relation induced by the rules \eqref{eq:comm}--\eqref{eq:app2}, \ie, $\equiv$ is the reflexive-transitive closure of $\equiv_1$. We then prove the following lemmata.
\begin{enumerate}
\item If $v\equiv_1 t$ and $v$ is a value, then $t$ is also a value.

(Remark that no structural equivalence rule allows for commuting an explicit substitution with a $\lambda$-abstraction. In the case $v=\pair{v_1,v_2}$, then the only possibility is $u=\pair{u_1,u_2}\alpha$ (rules \eqref{eq:app1}, \eqref{eq:app2}), but in this case we would have $v_i=u_i\alpha$, for $i=1$ or $i=2$, contradicting the hypothesis of $v_i$ to be a value).

\item If $\ctxt x\equiv_1 u$, then $u=\ctxt[\metaCtxt']{\subctxt x}$, for some context $\metaCtxt'$ and explicit substitution $\metaSubCtxt$.

(By induction on $\metaCtxt$ and case inspection).

\item If $t\equiv_1 t'$ and $t\red[\iota]u$, for some reduction step applied to the root of $t$, then there exists $u'\equiv u$ such that $t'\red[\iota] u'$.

(By case inspection. Notice that in case $t\esub{x}{v}\red[\eqref{red:gc}]t$ and $t\esub{x}{v}\equiv_1 t\esub{x}{u}$, with $v\equiv_1 u$, one should use (1) to assure that $u$ is a value, in particular no explicit substitution has been lifted from inside $v$).

\item If $t\equiv_1 t'$ and $t\red[\iota]u$, then there exists $u'\equiv u$ such that $t'\red[\iota] u'$.

(By 2 and 3).
\end{enumerate}

The general statement then follows by induction on the length of a sequence $t\equiv_1\dots\equiv_1 t'$ giving $t\equiv t'$, using (4) for the induction step. 
\end{proof}

\noindent\textsc{Proposition~\ref{prop:postponement}}
{\it Let $X$ be any subset of the reduction rules in \reftab{reductions} (including the variant $\eqref{red:subst}^{\mathfrak n}$) and let $t \mathrel{(\red[X]\cup\equiv)^\ast} u$ with $k$ $X$-steps, then there exists $u'$ such that $t\reds[k]{X} u'\equiv u$.}
\begin{proof}
By induction on $k$, using \reflemma{bisimulation}.
\end{proof}

\noindent\textsc{Proposition~\ref{prop:values}.}
{\it Given a closed term $t$ of type $A$, if $t$ is a $\beta$-normal form, then it is a value.}
\begin{proof}
By induction on a derivation $\vdash t:A$. If the last rule is an $\rightarrow$ introduction, then $t$ is a value. If it is an $\rightarrow$ elimination, then $t=u_1u_2$ and $\vdash u_1:B\rightarrow A$ or $\vdash u_1:\dual$. By IH $u_1$ is then a value, hence of the form $\lambda x.u'$ and therefore $u_1u_2$ is a redex \eqref{red:lambda}. If the last rule is a $\times$ introduction, we apply trivially the IH. If it is a $\times$ elimination, then $t=\unpair{t'}{x}{y}{u}$ and $\vdash u: A\times B$. By IH $u$ is a value, and hence it is of the form $u=\pair{v_1,v_2}$, therefore $t$ is a \eqref{red:pair} redex. The explicit substitution rule is similar to the previous one, using \eqref{red:subst} or \eqref{red:gc}. If $t=f(t_1,\dots,t_n)$, then we have by IH that each $t_i$ is a closed value, so a numeral, so we can apply the numerical rule associated with $f$. If $t$ is a numeral, then it is a value. If $t=t_1+t_2$, then by IH each $t_i$ is a value so we may apply either \eqref{red:plusPair} or \eqref{red:plus}.  
\end{proof}

\noindent\textsc{Proposition~\ref{prop:wn}}
{\it For every term $t$, and every set $X\subseteq\beta\ell$, there exists a $X$-normal form $u$ such that $t\reds{X} u$.}
\begin{proof}[Proof (Sketch)]
In case $\eqref{red:subst}\notin X$, it is trivial to find a notion of weight $w(t)\in\Nat$ such that $t\red[X]t'$ implies $w(t)>w(t')$. Otherwise, let $k$ be the greatest size of the type of a subterm of $t$, and let $W(t)$ be the sequence $(n_1,\dots, n_k, n_{k+1})$ where $n_{k+1}=w(t)$ and for every $1\leq i\leq k$, $n_i$ is the number of redexes of type $\eqref{red:lambda},\eqref{red:pair},\eqref{red:subst}$ in $t$ such that the term in the explicit substitution has a type of size $k+1-i$ (so that, as $i$ grows bigger, $n_i$ accounts for smaller and smaller types). It is not hard to show that, if $t$ is not $X$-normal, then $t\red[X]t'$ for some $t'$ such that $W(t)>W(t')$ in the lexicographical order. 
\end{proof}

\noindent\textsc{Lemma~\ref{lemma:size_shrinks}}
{\it For any $t\reds{\eqref{red:subst}^{\mathfrak n}\eqref{red:gc}\eqref{red:plus}\eqref{red:times}}u$, the number of reduction steps in the sequence is $O(\size t)$.}
\begin{proof}
Let $\bisize t$ be the number of symbols in $t$ where each variable is counted twice. The claim then follows by remarking that any rule \eqref{red:gc}, \eqref{red:plus}, \eqref{red:times} makes $\bisize t$ strictly decrease, and the same for \eqref{red:subst} whenever the substituted term is a numeral. 
\end{proof}

\begin{table}
\begin{subtable}{\textwidth}
\begin{align*}
\sem\tyR&:= \mathbb R\\
\sem{A\times B}&:= \sem{A}\times \sem{B}\\
\sem{A\rightarrow B}&:= \text{set of functions from }\sem{A}\text{ to }\sem{B}\\
\sem{\dual[d]}&:= \text{set of linear maps from }\mathbb R\text{ to }\sem{\tyR^d}\\
\sem{x_1^{\exptype{A_1}},\dots,x_n^{\exptype{A_n}}}&:= \sem{A_1}\times\dots\times\sem{A_n}
\end{align*}
\caption{interpretation of types}
\label{subtable:denotationTypes}
\end{subtable}

\begin{subtable}{\textwidth}
\begin{align*}
\sem{\Gamma\vdash t:A}:&\text{ function from }\sem{\Gamma}\text{ to } \sem{A}&
\sem{\Gamma\vdash_z t:\tyR^d}:&\text{ function from }\sem{\Gamma}\text{ to } \sem{\dual[d]}
\end{align*}

\begin{align}
\sem{\Gamma\vdash x:A}(\seq g)&:= g_x\\
\sem{\Gamma\vdash_z z:\tyR}(\seq g)&:= r\mapsto r\\
\sem{\Gamma\vdash f(\seq t):\tyR}(\seq g)&:=f(\sem{\seq{\Gamma\vdash t:\tyR}}(\seq g))\\
\sem{\Gamma\vdash\num r:\tyR}(\seq g)&:= r\\
\sem{\Gamma\vdash\lambda x^{\exptype{A}}. t:A\to B}(\seq g)&:= d \mapsto \sem{\Gamma,x^{\exptype{A}}\vdash t:B}(\seq g,d)\\
\sem{\Gamma\vdash tu:B}^\Gamma(\seq g)&:=\sem{\Gamma\vdash t:A\to B}(\seq g)(\sem{\Gamma\vdash u:A}(\seq g))\\
\sem{\Gamma\vdash\lambda z^{\lintype}. t:\dual[d]}(\seq g)&:= \sem{\Gamma\vdash_z t:\tyR^d}(\seq g)\\
\sem{\Gamma\vdash_z tu:\tyR^d}(\seq g)&:=a\mapsto\sem{\Gamma\vdash t :\tyR^d}(\seq g)(\sem{\Gamma\vdash_z u:\tyR}(\seq g)(a))
\label{sem:tu_lin}\\
\sem{\Gamma\vdash_z\pair{t, u}:\tyR^d\times\tyR^{d'}}^{\Gamma}(\seq g)&:=a\mapsto(\sem{\Gamma\vdash_z t:\tyR^d}(\seq g)(a),\sem{\Gamma\vdash_z u:\tyR^{d'}}(\seq g)(a))
\label{sem:pair_lin}\\
\sem{\Gamma\vdash_z t\esub{x^{\exptype A}}{u}:\tyR^d}(\seq g)&:=a\mapsto\sem{\Gamma,x^{\exptype A}\vdash t:\tyR^d}(\seq g,\sem{\Gamma\vdash u::\tyR}(\seq g))(a)
\label{sem:esub_lin}\\
\sem{\Gamma\vdash_{z_1} t\esub{z_2^{\lintype}}{u}:\tyR^d}(\seq g)&:=a\mapsto\sem{\Gamma\vdash_{z_2}t:\tyR^d}(\seq g)(\sem{\Gamma\vdash_{z_2}u:\tyR}(\seq g)(a))
\label{sem:esub_linlin}\\
\sem{\Gamma\vdash_z\unpair t {x^{\exptype A}} {y^{\exptype B}} u:\tyR^d}(\seq g)&:=a\mapsto\sem{\Gamma,x^{\exptype A},y^{\exptype B}\vdash_z t:\tyR^d}(\seq g,\sem{\Gamma\vdash u:A\times B}(\seq g))(a)\\
\sem{\Gamma\vdash_z t\cdot u:\tyR}(\seq g)&:=a\mapsto\sem{\Gamma\vdash_z t:\tyR}(\seq g)(a)\cdot\sem{\Gamma\vdash u:\tyR}(\seq g)\\
\sem{\Gamma\vdash_z\num 0:\tyR}(\seq g)&:=a\mapsto 0\\
\sem{\Gamma\vdash_zt+u:\tyR^d}(\seq g)&:=a\mapsto\sem{\Gamma\vdash_zt:\tyR^d}(\seq g)(a)+\sem{\Gamma\vdash_z u:\tyR^d}(\seq g)(a)
\end{align}
\caption{interpretation of terms. In the cases admitting two version (with or without linear variables), we considered only the case with linear variables, the other being simpler. The reader may  check that the linearity (i.e.~ commutation with sums and scalar products) of the parameter $a$ associated with a linear variable is respected.}
\label{subtable:denotationTerms}
\end{subtable}
\caption{The denotational model induced by the category of sets and functions.}
\label{table:denotational}
\end{table}
Table~\ref{table:denotational} recalls the denotational model induced by the cartesian closed structure of the category of sets and functions. The proof of the soundness property is completely standard.
\subsection{Proofs of Section~\ref{sect:results}}

\noindent\textsc{Lemma~\ref{lemma:backprop_size}.}
{\it Let $G$ be a ground term whose free variables are given by a sequence $\seq x$ of length $n$. Then, 
$$\size{\bp G}= O(\size G).$$
}
\begin{proof}
We actually prove a slight stronger claim: let $G$ be a ground term whose free variables are \emph{in} a sequence $\seq x$ of length $n$. Then, 
$$\size{\bp G}= O(n+\size G).$$
Notice that then whenever $\seq x$ contains exactly the free variables of $G$, $n\leq \size G$ and so $O(n+\size G)=O(\size G)$. 

By induction on the definition of $\bp G=\pair{G_0,\pair{G_1,\ldots,G_n}\alpha}\beta$, one proves that:
\begin{itemize}
\item[(i)] $\size{G_0\beta}= O(\size G)$;
\item[(ii)] for all $0<i\leq n$, $\size{G_i}= O(k(\#_{x_i}G +1))$, with $\#_{x_i}G$ is the number of occurrences of $x_i$ in $G$;
\item[(iii)] for all $0<i\leq n$, $\size{G_i\alpha}= O(\size G)$.
\end{itemize}
Notice that the items (i)-(iii) give that 
$\size{\bp G}=\sum_{i=1}^n\size{G_i}+\size{G_0\beta}+\size{\alpha}= O(\#_{x_1}G+\dots+\#_{x_n}G)+2\size G) =O(\size G))$. 

We detail only the case $G=F\esub{z}{F'}$, taking the notation of the definition of $\bp[]{\,}$. Item (i) follows trivially from the induction hypothesis. Item (ii) is a consequence of the induction hypothesis and the fact that $\#_{x_i}(G)=\#_{x_i}(F)+\#_{x_i}(F')$.
Concerning item (iii), we have that: $\size{(F_i\oplus F_i')\alpha'\esub{b}{H}\alpha} = \size{F_i\alpha}+\size{F_i'\alpha'}+1+\size H= O(\size F+\size{F'})= O(\size G)$.

The factor $k$ in (ii) is due to the case $G=f(x_{i_1},\ldots,x_{i_k'})$.
\end{proof}
Notice that the above lemma would fail if we replace the meta-operator $\oplus$ with the sum constructor $+$ in the definition of $\bp{G}$. In that case we would have the size of $\bp{x_1+\dots+x_n}$ to be quadratic in $n$ because $\bp{x_1+\dots+x_n}$ copies the $n-1$ sums for every component of the gradient tuple. 

\noindent\textsc{Proposition \ref{prop:backprop_sound}.}
{\it Let $G$ be a ground term whose free variables are given by a sequence $\seq x$ of length $n$. Then for every $\seq r\in \mathbb R^n$, we have:
	$$
		\bp{G}\esub{a}{\num 1}\esub{\seq x}{\seq{\num r}}
		\;\reds[O(\size G)]{\eqref{red:subst}^{\mathfrak n}\eqref{red:gc}\eqref{red:plus}\eqref{red:times}}\equiv\;
		\pair{\num{\sem{G}(\seq{r})},\nabla G(\seq r)}.
	$$
}
\begin{proof}
By induction on $G$ we prove that, if $\bp{G}=\pair{G_0,\pair{G_1,\ldots,G_n}\alpha}\beta$, then, for every vector of real numbers $\seq r$, we have:
	\begin{itemize}
	\item[(i)] $G_0\beta\esub{\seq x}{\seq{\num r}}\reds{\eqref{red:subst}^{\mathfrak n}\eqref{red:gc}\eqref{red:plus}\eqref{red:times}}\equiv \num{\sem{G}(\seq{r})}$;
	\item[(ii)]  for all $1\leq i\leq n$, for all real number $q$, $G_i\alpha\beta\esub{a}{\num q}\esub{\seq x}{\seq{\num r}}\reds{\eqref{red:subst}^{\mathfrak n} \eqref{red:gc}\eqref{red:plus}\eqref{red:times}}\equiv\num{\partial_{x_i} \sem{G}(\seq r)\cdot q}$;
	\end{itemize}
From $(i)$ and $(ii)$ then follows: 
\begin{multline*}
\bp{G}\esub{a}{\num 1}\esub{\seq x}{\num r}\equiv\\
\pair{G_0\beta\esub{\seq x}{\seq{\num r}},\pair{G_1\alpha\beta\esub{a}{\num 1}\esub{\seq x}{\seq{\num r}},\ldots,G_n\alpha\beta\esub{a}{\num 1}\esub{\seq x}{\seq{\num r}}}}
\reds{\eqref{red:subst}^{\mathfrak n}\eqref{red:gc}\eqref{red:plus}\eqref{red:times}}\equiv
\pair{\num{\sem{G}(\seq{r})},\nabla G(\seq r)}.
\end{multline*}
Proposition~\ref{prop:postponement} allows to postpone all structural equivalences at the end. From Lemma~\ref{lemma:size_shrinks} and the fact that $\size{\bp{G}}=O(\size G)$ (Lemma~\ref{lemma:backprop_size}) the length of the reduction is $O(\size G)$

The non-trivial case for $(i)$ and $(ii)$ is $G=F'\esub{z}{F''}$. Using the notations from the definition:
\begin{align*}
G_0\beta\esub{\seq x}{\seq{\num r}}
&=F_0'\beta'\esub{z}{F_0''}\beta''\esub{\seq x}{\seq{\num r}}\\
&\equiv F_0'\beta'\esub{\seq x}{\seq{\num r}}\esub{z}{F_0''\beta''\esub{\seq x}{\seq{\num r}}}&
\\
&\reds{\eqref{red:subst}^{\mathfrak n}\eqref{red:gc}\eqref{red:plus}\eqref{red:times}}
F_0'\beta'\esub{\seq x}{\seq{\num r}}\esub{z}{\num{\sem{F''}(\seq r)}}&\text{by IH}\\
&
\reds{\eqref{red:subst}^{\mathfrak n}\eqref{red:gc}\eqref{red:plus}\eqref{red:times}}
\num{\sem{F'}(\seq{ r},\sem{F''}(\seq r))}=\num{\sem{G}(\seq{ r})}&\text{by IH}
\end{align*} 

As for (ii), we have that the $i$-th component of the gradient tuple of $\bp{G}$ together with the associated substitutions is equal to (we suppose $F_i'\oplus F_i''=F_i'+F_i''$, the other cases being simpler):
\begin{align*}
&(F_i'+F_i'')\alpha''\esub{b}{H}\alpha'\beta'\esub{z}{F_0''}\beta''\esub{a}{\num q}\esub{\seq x}{\seq{\num r}}\\
&\equiv(F_i'+F_i'')\alpha''\esub{b}{H}\alpha'\beta'\beta''\esub{a}{\num q}\esub{\seq x}{\seq{\num r}}\esub{z}{F_0''\beta''\esub{\seq x}{\seq{\num r}}}&
\\
&
\reds{\eqref{red:subst}^{\mathfrak n}\eqref{red:gc}\eqref{red:plus}\eqref{red:times}}
(F_i'+F_i'')\alpha''\esub{b}{H}\alpha'\beta'\beta''\esub{a}{\num q}\esub{\seq x}{\seq{\num r}}\esub{z}{\num{\sem{F''}(\seq r)}}&\text{by (i)}\\
&\equiv(F_i'+F_i'')\alpha''\alpha'\beta'\beta''\esub{a}{\num q}\esub{\seq x}{\seq{\num r}}\esub{z}{\num{\sem{F''}(\seq r)}}\esub{b}{H\alpha'\beta'\beta''\esub{a}{\num q}\esub{\seq x}{\seq{\num r}}\esub{z}{\num{\sem{F''}(\seq r)}}}&
\\
&\reds[O(\size F)]{\eqref{red:subst}^{\mathfrak n}\eqref{red:gc}\eqref{red:plus}\eqref{red:times}}
(F_i'+F_i'')\alpha''\alpha'\beta'\beta''\esub{a}{\num q}\esub{\seq x}{\seq{\num r}}\esub{z}{\num{\sem{F''}(\seq r)}}\esub{b}{\num{\partial_z \sem{F'}(\seq r)\cdot q}}&\text{by IH}\\
&\equiv 
F_i'\alpha'\beta'\esub{a}{\num q}\esub{\seq x}{\seq{\num r}}\esub{z}{\num{\sem{F''}(\seq r)}} 
+
F_i''\alpha''\beta''\esub{b}{\num{\partial_z \sem{F'}(\seq r)\cdot q}}\esub{\seq x}{\seq{\num r}}&
\\
&
\reds{\eqref{red:subst}^{\mathfrak n}\eqref{red:gc}\eqref{red:plus}\eqref{red:times}}
\num{\partial_{x_i}\sem{F'}(\seq r)\cdot q}
+
\num{\partial_{x_i}\sem{F''}(\seq r)\cdot \left(\partial_z \sem{F'}(\seq r)\cdot q\right)}&\text{by IH}\\
&\red[\eqref{red:plus}]\num{\partial_{x_i}\sem{F'}(\seq r)\cdot q+\partial_{x_i}\sem{F''}(\seq r)\cdot \left(\partial_z \sem{F'}(\seq r)\cdot q\right)}\\
&=\num{
	\left({\partial_{x_i}\sem{F'}}(\seq r)+\partial_z \sem{F'}(\seq r)\cdot \partial_{x_i}\sem{F''}(\seq r)\right)\cdot q
}=\num{\partial_{x_i}\sem{G}(\seq r)\cdot q}
\end{align*}
Notice that in order to move to the last line we  use the associativity, commutativity and distributivity of $+$ and $\cdot$ over real numbers, not on the corresponding syntactic symbols. 

Let us consider also the case $G=\pair{F',F''}$. Using the notations from the definition:
\begin{align*}
G_0\beta\esub{\seq x}{\seq{\num r}}
&=\pair{F_0',F_0''}\beta'\beta''\esub{\seq x}{\seq{\num r}}\\
&\equiv
\pair{F_0'\beta'\esub{\seq x}{\seq{\num r}},F_0''\beta''\esub{\seq x}{\seq{\num r}}}&
\\
&\reds{\eqref{red:subst}^{\mathfrak n}\eqref{red:gc}\eqref{red:plus}\eqref{red:times}}\,\equiv\,
\pair{\num{\sem{F'}(\seq r)},\num{\sem{F''}(\seq r)}}&\text{by IH}\\
&=\num{\sem{G}(\seq{ r})}
\end{align*} 
The case (ii) is similar (just a sum instead of a pair). 

In the base cases (variables, functional symbols and numerals), one performs linear substitutions~\eqref{red:subst} as well as garbage collection~\eqref{red:gc}. In the case of $\bp[\seq x,a]{f(x_{i_1},\dots,x_{i_k}})$  one instance of the rule~\eqref{red:times} is needed.
\end{proof}

\begin{lemma}\label{lemma:sizeReverse}
Let $t$ be a term of $\Terms[\mathcal F]$. We have: $\size{\revder{t}}= O(\size t)$.
\end{lemma}
\begin{proof}
By just inspecting the definition in Table \ref{table:reverse}.
\end{proof}

\noindent\textsc{Lemma~\ref{lemma:Commutation}.}
{\it	Let $t$ be a term of $\Terms[\mathcal F]$. Then:
	\begin{enumerate}
		\item if $t\equiv t'$, then $\revder{t}\equiv\revder{t'}$,
		\item if $t\red[\iota]t'$, for $\iota$ any reduction step in \reftab{reductions}, then $\revder{t}\reds[O(1)]{X}\revder{t'}$, where $X=\{\iota\}$ for any $\iota$ but  $\eqref{red:plus}, \eqref{red:times}$, in these latter cases $X=\{\iota, \eqref{red:lambda}, \eqref{red:subst}, \eqref{red:gc}\}$.		
	\end{enumerate}
}
\begin{proof}
Item 1 is immediate by case inspection. Concerning item 2, one proves by induction on $t$ that:
\begin{itemize}
\item[($\bullet$)] for every term $u$ of $\Terms[\mathcal F]$, $\revder{t\isub{u}{x}}=\revder{t}\isub{\revder u}{x}$.
\end{itemize}
Then, suppose $t\red[\iota]t'$, notice that $\iota\neq \eqref{red:linear} $ since $t$ has no $\dual$-variable. If $\iota=\eqref{red:subst}$, then the statement follows from ($\bullet$). All other cases are immediate. We detail just the case of an instance of a numerical rule $\eqref{red:plus}, \eqref{red:times}$, e.g.~ $\num r+\num q\red\num{r+q}$. We have that $\revder{\num r+\num q}$ is:
\begin{multline*}
	\pair{x+y,\lambda a. (x'a+y'a)}\esub{\pair{x,x'}}{\pair{\num r, \lambda a.\num 0}}\esub{\pair{y,y'}}{\pair{\num q,\lambda a.\num 0}}\\
	\reds[O(1)]{\eqref{red:lambda}\eqref{red:subst}\eqref{red:gc}}{}\pair{\num r+\num q,\lambda a. (\num 0+\num 0)}\reds[2]{\eqref{red:plus}}\pair{\num{r+q},\lambda a. \num 0}
\end{multline*}
\end{proof}

\noindent\textsc{Lemma~\ref{lemma:main}.}
{\it
Let $G$ be a ground term of type $\seq x^{\exptype{\tyR}}\vdash G:\tyR$ with $\seq x=x_1^{\exptype{\tyR}},\dots,x_n^{\exptype{\tyR}}$ the free variables of $G$ and let $\bp{G}=\pair{G_0,\pair{G_1,\ldots,G_n}\alpha}\beta$. 
	For a suitable $J\subseteq \{1,\dots,n\}$, such that if $i\notin J$, then $G_i= \num 0$, we have:
		$$\revder{G}\esub{\seq x^{\exptype{\revder\tyR}}}{\pair{\seq x^{\exptype{\tyR}},\lambda a^\tyR.{\seq x'}^{\exptype{\dual}} a}}\freds[O(\size G)]\equiv
		\pair{G_0,\lambda a.(\sum_{j\in J} x_j'G_{j})\alpha}\beta.$$
\begin{proof}
By induction on $G$. We will freely use structural equivalence where necessary, knowing that Proposition~\ref{prop:postponement} allows us to postpone all structural equivalences to the end without affecting the number of reduction steps. The non-trivial cases are detailed below. Remark that the length of each reduction is $O(\size{\revder{G}})$ which is equal to $O(\size G)$ by Lemma~\ref{lemma:sizeReverse}.

Let $G=F\esub{z}{F'}$ and suppose that
	\begin{align*}
		\bp[\seq x,z,a]{F} &= \pair{F_0,\pair{F_1,\ldots,F_n,H}\alpha}\beta, &
		\bp[\seq x,b]{F'} &= \pair{F_0',\pair{F_1',\ldots,F_n'}\alpha'}\beta'.
	\end{align*}
	We have that $\revder{G}\esub{\seq x}{\pair{\seq x,\lambda a.\seq x'a}}$ is structural equivalent to  :
	\begin{align*}
	&(\revder{F}\esub{\seq x}{\pair{\seq x,\lambda a.\seq x'a}})\esub{z}{(\revder{F'}\esub{\seq x}{\pair{\seq x,\lambda b.\seq x'b}})}
	&
	\\
	&
	\freds[O(\size{F'})]\equiv
	(\revder{F}\esub{\seq x}{\pair{\seq x,\lambda a.\seq x'a}})\esub{z}{\pair{F_0',\lambda b.(\sum_{j\in J'} x_j'F'_{j})\alpha'}\beta'}
	&\text{by IH}
	\\
	&\equiv(\revder{F}\esub{\seq x}{\pair{\seq x,\lambda a.\seq x'a}})\esub{z}{\pair{F_0',\lambda b.(\sum_{j\in J'} x_j'F'_{j})\alpha'}}\beta'&
	\\
	&\red[\eqref{red:etaPair}]\red[\eqref{red:pair}](\revder{F}\esub{\seq x}{\pair{\seq x,\lambda a.\seq x'a}}\esub{z}{\pair{z,z'}})\esub{z}{F_0'}\esub{z'}{\lambda b.(\sum_{j\in J'} x_j'F'_{j})\alpha'}\beta'
	\\
	&\red[\eqref{red:etaLambda}](\revder{F}\esub{\seq x}{\pair{\seq x,\lambda a.\seq x'a}}\esub{z}{\pair{z,\lambda a.z'a}})\esub{z}{F_0'}\esub{z'}{\lambda b.(\sum_{j\in J'} x_j'F'_{j})\alpha'}\beta'
	\end{align*}
	
	Now we can apply the induction hypothesis on $\revder{F}$ and we split in two sub-cases, depending whether the variable $z'$ appears in the tuple associated with  $\revder{F}$ or not. In the first case, we have that the above term reduces by IH to:
	\begin{align*}
	&\freds[O(\size F)]\equiv
	\pair{F_0,\lambda a.(\sum_{j\in J} x_j'F_{j}+z'H)\alpha}\beta\esub{z}{F_0'}\esub{z'}{\lambda b.(\sum_{j\in J'} x_j'F'_{j})\alpha'}\beta'
	&\text{by IH}
	\\
	&\red[\eqref{red:subst}]\red[\eqref{red:gc}]\pair{F_0,\lambda a.(\sum_{j\in J} x_j'F_{j}+(\lambda b.(\sum_{j\in J'} x_j'F'_{j})\alpha')H)\alpha}\beta\esub{z}{F_0'}\beta'\\
	&\red[\eqref{red:lambda}]\pair{F_0,\lambda a.(\sum_{j\in J} x_j'F_{j}+(\sum_{j\in J'} x_j'F'_{j})\alpha'\esub{b}H)\alpha}\beta\esub{z}{F_0'}\beta'\\
	&
	\reds[\#J\cap\#J']{\eqref{red:linear}}
	\pair{F_0,\lambda a.(\sum_{j\in J\cup J'} x_j'(F_{j}\oplus F'_{j}))\alpha'\esub{b}H\alpha}\beta\esub{z}{F_0'}\beta'
	\end{align*}
	Notice that the above instance of~\eqref{red:subst} replaces exactly one occurrence of $z'$. 
The last term satisfies the statement of the lemma.

	In the case the variable $z'$ does not appear in the tuple associated with  $\revder{F}$, this means that $H = \num 0$. We then have: 
	\begin{align*}
	&\reds[O(\size F)]{\beta\eta\ell}\equiv
	\pair{F_0,\lambda a.(\sum_{j\in J} x_j'F_{j})\alpha}\beta\esub{z}{F_0'}\esub{z'}{\lambda b.(\sum_{j\in J'} x_j'F'_{j})\alpha'}\beta'
	&\text{by IH}
	\\
	&\red\pair{F_0,\lambda a.(\sum_{j\in J} x_j'F_{j})\alpha}\beta\esub{z}{F_0'}\beta'	
	&\text{by~\eqref{red:gc}}
	\end{align*}	
The last term satisfies the statement of the lemma.

Let $G=f(x_{i_1},\ldots,x_{i_k})$ for a suitable subset $x_{i_1},\ldots,x_{i_k}\subseteq\seq x$. In this case we have: 
	\begin{multline*}
		\revder{G}\esub{\seq x}{\pair{\seq x,\lambda a.\seq x'a}}=\unpair{\pair{f(\seq y)\ ,\ \lambda a.\sum_{j=1}^k y_j'\left({\partial_j f}(\seq y)\cdot a\right)}}
		{\seq y}{\seq y'}{\seq x}\esub{\seq x}{\pair{\seq x,\lambda a.\seq x'a}}\\
		\reds[O(n)]{\eqref{red:pair}\eqref{red:subst}} 
		\pair{f(\seq y)\ ,\ \lambda a.\sum_{j=1}^k y_j'\left({\partial_j f}(\seq y)\cdot a\right)}
		\esub{\seq y}{\seq x}\esub{\seq y'}{\lambda a.\seq x'a}\\
		\equiv
		\pair{f(\seq y)\ ,\ \lambda a.\sum_{j=1}^k y_j'\left({\partial_j f}(\seq{\overline y})\cdot a\right)}
		\esub{\seq y}{\seq x}\esub{\seq{\overline y}}{\seq x}\esub{\seq y'}{\lambda a.\seq x'a}\\
		\reds[O(k)]{\eqref{red:lambda}\eqref{red:subst}} 
		\reds[O(n)]{\eqref{red:gc}} 
		\pair{f(x_{i_1},\ldots,x_{i_k})\ ,\ \lambda a.\sum_{j=1}^k x_{i_j}'\left(\partial_j f(x_{i_1},\ldots,x_{i_k})\cdot a\right)}
	\end{multline*}
	\item Let $G=F+F'$ and  suppose that
	\begin{align*}
		\bp[\seq x,a]{F} &= \pair{F_0,\pair{F_1,\ldots,F_n}\alpha}\beta, &
		\bp[\seq x,a]{F'} &= \pair{F_0',\pair{F_1',\ldots,F_n'}\alpha'}\beta'.
	\end{align*}
		We have that $\revder{G}\esub{\seq x}{\pair{\seq x,\lambda a.\seq x'a}}$ is equal to:
	\begin{align*}
		&\pair{y_1+y_2,\lambda a. (y_1'a+y_2'a)}\esub{\pair{y_1,y_1'}}{\revder F}\esub{\pair{y_2,y_2'}}{\revder{F'}}\esub{\seq x}{\pair{\seq x,\lambda a.\seq x'a}}\\
		&
		\reds[O(\size G)]{\beta\eta\ell}\equiv
		\pair{F_0+F_0',\lambda a. (y_1'a+y_2'a)}\esub{y_1'}{
		\lambda a.(\sum_{j\in J} x_j'F_{j})\alpha}\beta
		\esub{y_2'}{
		\lambda a.(\sum_{j'\in J'} x_{j'}'F'_{j'})\alpha'}\beta'
		&&\text{by IH}\\
		&
		\reds[O(1)]{\eqref{red:subst}\eqref{red:lambda}\eqref{red:gc}}
		\pair{F_0+F_0',\lambda a. ((\sum_{j\in J} x_j'F_{j})\alpha+(\sum_{j'\in J'} x_{j'}'F'_{j'})\alpha')}\beta\beta'	
		\\
		&\equiv\pair{F_0+F_0',\lambda a. ((\sum_{j\in J} x_j'F_{j})+(\sum_{j'\in J'} x_{j'}'F'_{j'}))\alpha\alpha'}\beta\beta'	
		&&
		\\
		&\reds[\#J\cap\#J']{\eqref{red:linear}}\pair{F_0+F_0',\lambda a. (\sum_{j\in J\cup J'} x_j'F_{j}\oplus F'_{j'})\alpha\alpha'}\beta\beta'
	\end{align*}
	Notice that all the above instances of~\eqref{red:subst} replaces exactly one occurrence of a variable. 
	The last term satisfies the statement of the lemma.
\end{proof}

\noindent\textsc{Theorem~\ref{th:main theorem}.}
{\it Let $t$ be a term of $\Terms[\mathcal F]$ of type $\seq x^{\exptype{\tyR}}\vdash t: \tyR$, with $\seq x=x_1^{\exptype{\tyR}},\dots,x_n^{\exptype{\tyR}}$ the free variables of $t$. For any ground term $G$ such that $t\, (\red[\beta]\cup\equiv)^*\, G$ in $m$ $\beta$-steps, we have, for a suitable $J\subseteq\{1,\dots,n\}$:
\[
	\revder t\esub{\seq x}{\pair{\seq x,\lambda a.\seq x'a}} \freds[O(m+\size G)]\,\equiv \pair{G_0,\lambda a.(\sum_{j\in J} x_j'G_{j})\alpha}\beta
\]
\noindent where $\bp{G}=\pair{G_0,\pair{G_1,\ldots,G_n}\alpha}\beta$ and $\forall i\notin J$, $G_i=\num 0$. 
}
\begin{proof}
Let us suppose $t\mathrel{(\red[\beta]\cup\equiv)^\ast}G$ in $m$ steps. By Lemma~\ref{lemma:Commutation}, we have that: 
$$
	\revder t\esub{\seq x}{\pair{\seq x,\lambda a.\seq x'a}}\mathrel{(\red[\beta]\cup\equiv)^\ast} \revder G\esub{\seq x}{\pair{\seq x,\lambda a.\seq x'a}}
$$ 
in $O(m)$ steps. By \reflemma{main}, we have that: 
$$
	\revder G\esub{\seq x}{\pair{\seq x,\lambda a.\seq x'a}}\freds[O(\size G))]\equiv\pair{G_0,\lambda a.(\sum_{j\in J} x_j'G_{j})\alpha}\beta
$$ 
this latter term satisfying the statement of the theorem. We compose the reductions and apply \refprop{postponement} to conclude.
\end{proof}

\noindent\textsc{Corollary~\ref{cor:gradient}.}
{\it Let $t$ be a term of $\Terms[\mathcal F]$ of type $\seq x^{\exptype{\tyR}}\vdash t: \tyR$, with $\seq x=x_1^{\exptype{\tyR}},\dots,x_n^{\exptype{\tyR}}$ the free variables of $t$.  Let  $m$ be the number of $\beta$-steps needed to reduce $t$ to a ground term $G$.
For any vector $\seq r\in\mathbb R^n$, let $\nabla t(\seq r)=\pair{\num{g_1},\dots,\num{g_n}}$. Then, for a suitable $J\subseteq\{1,\dots,n\}$, with $i\notin J$, $g_i=0$, we have:
\[
	\unpair{z'\num 1}z{z'}{\revder t\esub{\seq x}{\pair{\seq x,\lambda a.\seq x'a}}}\esub{\seq x}{\num{\seq r}}\freds[O(m+\size G)]\,\equiv\sum_{j\in J}x_j'\num{g_j}.
\]
}
\begin{proof}
Apply Theorem \ref{th:main theorem}  to $\revder t\esub{\seq x}{\pair{\seq x,\lambda a.\seq x'a}}$ and then Proposition \ref{prop:backprop_sound}. 
\end{proof}

\subsection{Proofs of Section~\ref{sect:example}}

\begin{lemma}\label{lemma:example-backprop}
  Let $r \in \Real$ and $u,u' \in \Terms[\mathcal F]$.
  $\revder{L}\revder{\num{r}}\pair{u,u'}$ reduces to
  $$\unpair{\unpair{\pair{\sigma(z_1.\num{r} + z_2.u), \lambda a.(z_1'(\num{r}.\sigma'.a)  +  z_2'(u.\sigma'.a)  +  u'(z_2.\sigma'.a))}}{z_1}{z_1'}{\epsilon}}{z_2}{z_2'}{\rho}$$
  where $\sigma'=  \partial_1 \sigma(z_1.\num{r} + z_2.u)$.
\end{lemma}
\begin{proof}
  The proof follows from the following observations:
  \begin{align*}
    \revder{\epsilon.x} & = \unpair{\unpair{\pair{z_1.e, \lambda a(z_1'(e.a) + e'(z_1.a))}}{z_1}{z_1'}{\epsilon}}{e}{e'}{x} & \\
    \revder{\rho.h} & = \unpair{\unpair{\pair{z_2.e, \lambda a(z_2'(e.a) + e'(z_2.a))}}{z_2}{z_2'}{\rho}}{e}{e'}{h} &\\
    \revder{\epsilon.x + {\rho}.h} & = \unpair{\unpair{\pair{\gamma_1 + \gamma_2, \lambda a.(\gamma_1' a + \gamma_2' a)}}{\gamma_1}{\gamma_1'}{\revder{\epsilon.x}}}{\gamma_2}{\gamma_2'}{\revder{\rho.h}} &\\
    \revder{L} & = \lambda x.\lambda h.\unpair{\pair{\sigma(\mu), \lambda a.\mu' ({\partial_1 \sigma}(\mu).a)}}{\mu}{\mu'}{\revder{\epsilon.x+\rho.h}}&
  \end{align*}
\end{proof}

\begin{lemma}\label{lemma:recurrent-gradient}
  Let $a_1,\dots, a_n : \tyR$, $e,r\in \Real$ and $1\leq i \leq n$. We pose
  \begin{align*}
	\alpha&  = \esub{\epsilon}{\pair{\num{e}, \lambda a.\epsilon'a}}\esub{\rho}{\pair{\num{r}, \lambda \rho'a}}\\
	l_0 & := []\\
	l_{i+1} & := [\num{a_{n-i}}, \dots, \num{a_n}] 
  \end{align*}
  Then
  $\revder{Nl_i}\alpha$ reduces to $\pair{u_i, u_i'}$ where:
  \begin{align*}
    \pair{u_0,u_0'} & = \pair{\num{0},\lambda a.\num{0}} \\
    \pair{u_{i+1}, u_{i+1}'} & = \pair{\sigma(\num{e}.\num{a_{i+1}} + \num{r}.u_i), \lambda a.(\epsilon'(\num{a_{i+1}}.\sigma'_{i+1}.a)  +  \rho'(u_i.\sigma'_{i+1}.a) + u_i'(\num{r}.\sigma'_{i+1}.a))}
  \end{align*}
  with $\sigma'_{i+1}=  {\partial_1 \sigma}(\num{e}.\num{a_{i+1}} + \num{r}.u_i)$.
\end{lemma}
\begin{proof}
We have:
$$\revder{N[\num{a_{n-i}}; \dots; \num{a_n}]} \reds{} 
\revder{L}\revder{\num{a_{n-i}}}(\revder{L}\revder{\num{a_{n-i+1}}}( \dots (\revder{L}\revder{\num{a_n}}\pair{\num{0},\lambda a.\num{0}}) \dots))$$
Then the result follows by applying recursively Lemma \ref{lemma:example-backprop}.
\end{proof}

\noindent\textsc{Proposition~\ref{prop:example-grad}.}
{We have $\nabla (Nl)(e, r) = \pair{g^n_{\epsilon}, g^n_{\rho}}$ where $g^n_\epsilon$ and $g^n_\rho$ are given by the following recurrent equations:
\begin{align*}
  g^0_{\epsilon} & = \num{0} & g^{i+1}_{\epsilon} & = \sigma_{i+1}'.(\num{a_{i+1}} + \num{r}.g^{i}_{\epsilon})& \sigma'_{i+1} & :=  {\partial_1 \sigma}(\num{e}.\num{a_{i+1}} + \num{r}.u_i) \\
  g^0_{\rho} &= \num{0} & g^{i+1}_{\rho} &=  \sigma_{i+1}'.(u_{i+1} + \num{r}.g^{i}_{\rho}) & & \\
  u_0 &= \num{0} & u_{i+1} &= \sigma(\num{e}.\num{a_{i+1}} + \num{r}.u_{i}) & &
\end{align*}
}
\begin{proof}
By Corollary \ref{cor:gradient}, this amounts to computing $D = \epsilon'.g_\epsilon + \rho'.g_\rho$ such that $$
\unpair{z'\num 1}{z}{z'}{\revder{Nl}\esub{\epsilon}{\pair{\epsilon,\lambda a.\epsilon'a}}\esub{\rho}{\pair{\rho,\lambda a.\rho'a}}}\esub{\epsilon}{\num{e}}\esub{\rho}{\num{r}} \reds{} D
$$
We get the result by direct application of Lemma \ref{lemma:recurrent-gradient}.
\end{proof}
\end{document}